\pdfoutput=1
\documentclass[journal]{IEEEtran}

\usepackage{subfig}

\usepackage{amssymb}
\usepackage{amsmath}

\usepackage{mathabx}

\usepackage{amsthm}
\usepackage{tikz}

\usepackage{graphicx}

\newtheorem{theorem}{\textbf{Theorem}}

\newcommand*\circled[1]{\tikz[baseline=(char.base)]{
\node[fill = white, shape=circle,draw, minimum width=10pt, inner sep=0pt] (char) {#1};}}

\usepackage{booktabs}    
\usepackage{tabularx}
\usepackage{multirow}
\newcolumntype{L}[1]{>{\raggedright\let\newline\\\arraybackslash\hspace{0pt}}m{#1}}
\newcolumntype{C}[1]{>{\centering\let\newline\\\arraybackslash\hspace{0pt}}m{#1}}
\newcolumntype{R}[1]{>{\raggedleft\let\newline\\\arraybackslash\hspace{0pt}}m{#1}}

\usepackage[normalem]{ulem}

\usepackage[ruled,vlined,linesnumbered]{algorithm2e}
\usepackage{listings}
\usepackage{paralist}
\usepackage{soul}

\usepackage{hyperref}
\hypersetup{
unicode=true,
pdfstartview={FitH},
colorlinks=true,
citecolor=blue,
linkcolor=blue,
}
\usepackage{enumitem}

\begin{document}
\bstctlcite{IEEEexample:BSTcontrol}

\newcommand{\papertitle}[0]{ASSURE: RTL Locking Against an Untrusted Foundry}

\markboth{Accepted for publication in IEEE Transactions on Very Large Scale Integration (TVSLI) Systems, April~2021}{C. 
Pilato \MakeLowercase{\textit{et al.}}: ASSURE: RTL Locking Against an Untrusted Foundry}

\title{ASSURE: RTL Locking Against an\\ Untrusted Foundry}

\author{
Christian~Pilato,~\IEEEmembership{Senior~Member,~IEEE,}
Animesh~Basak~Chowdhury,~\IEEEmembership{Student~Member,~IEEE,}\\ Donatella~Sciuto,~\IEEEmembership{Fellow,~IEEE,}
Siddharth~Garg,~\IEEEmembership{Member,~IEEE,}
Ramesh~Karri,~\IEEEmembership{Fellow,~IEEE}
\thanks{Manuscript received October 11, 2020; revised January 28, 2021 and March 21, 2021; accepted April 6, 2021.
 
C. Pilato and D. Sciuto are with the Dipartimento di Elettronica, 
Informazione e Bioingegneria, Politecnico di Milano, Milano, Italy (christian.pilato@polimi.it, donatella.sciuto@polimi.it). 

A. B. Chowdhury, S. Garg and R. Karri are with the NYU Center for Cybersecurity 
(\url{http://cyber.nyu.edu}), New York University, New York, NY, USA (abc586@nyu.edu, garg@nyu.edu, rkarri@nyu.edu).
}}

\maketitle

\begin{abstract}
Semiconductor design companies are integrating proprietary intellectual property (IP) blocks to build custom integrated circuits (IC) and fabricate them in a third-party foundry. Unauthorized IC copies cost these companies billions of dollars annually. While several methods have been proposed for hardware IP obfuscation, they operate on the gate-level netlist, i.e., after the synthesis tools embed most of the semantic information into the netlist.  We propose ASSURE to protect hardware IP modules operating on the register-transfer level (RTL) description. The RTL approach has three advantages: (i) it allows designers to obfuscate IP cores generated with many different methods (e.g., hardware generators, high-level synthesis tools, and pre-existing IPs); (ii) it obfuscates the semantics of an IC before logic synthesis; (iii) it does not require modifications to EDA flows. We perform a cost and security assessment of ASSURE against state-of-the-art oracle-less attacks.
\end{abstract}

\IEEEpeerreviewmaketitle

\section{Introduction}

The cost of IC manufacturing has increased 5$\times$ when scaling from 90nm to 7nm~\cite{iccost}. An increasing number of design houses are now {\em fab-less} and outsource the fabrication to a third-party foundry~\cite{fabless,manufacturing}. 
This reduces the cost of operating expensive foundries but raises security issues. If a rogue in the third-party foundry has access to the design files, they can reverse engineer the IC functionality to steal the Intellectual Property (IP), causing economic harm to the design house~\cite{pieee14}. 

Fig.~\ref{fig:design_flow} is a fabless IC design flow with third-party manufacturing. The flow accepts the specification in a {\em hardware description language} (HDL). Designers create the components either manually or generate them automatically, and integrate them into a hardware description at the register-transfer level (RTL).  Given a technology library (i.e., a description of gates in the target technology) and a set of constraints, logic synthesis elaborates the RTL into  a gate-level netlist. Logic synthesis applies optimizations to reduce area and improve timing. 
While RTL descriptions are hard to match against high-level specifications~\cite{8000621}, they are used as a golden reference during synthesis to verify each step does not introduce any error. Physical design generates the layout files that are sent to the foundry for fabrication of ICs that are then returned to the design house for packaging and testing. Assuming a trusted design house, the foundry is the first place where a malicious attacker can reverse engineer and replicate an IC.

\begin{figure}[!tp]
\centering
\includegraphics[width=\columnwidth]{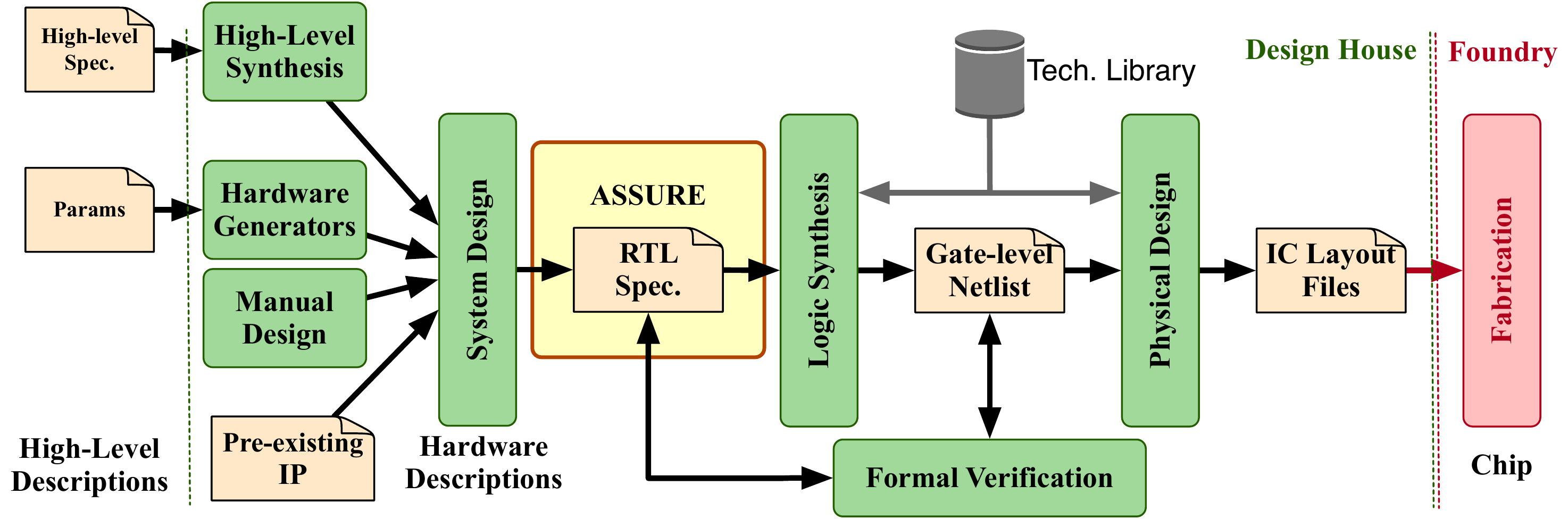}
\caption{State-of-the-art IC design flow. 
Designers create RTL description of an IC either by  manual design, or by using HLS tools, or by using hardware generators. The netlist after more processing steps is sent to a third-party foundry. ASSURE locks an RTL description before logic synthesis. }
\label{fig:design_flow}	
\end{figure}

Semiconductor companies are developing methods for IP obfuscation. 
In {\em split manufacturing}, the design house splits the IC into parts that are fabricated by different foundries~\cite{6513707}. An attacker must access all parts to recover the IC. While the design process becomes more complex, the designers cannot guarantee complete security. 
{\em Watermarking} hides a signature inside the circuit, which is later verified during litigation~\cite{abdel2003ip}. Finally, designers apply {\em logic locking}~\cite{10.1145/3342099} to prevent unauthorized copying and thwart reverse-engineering. They introduce extra gates controlled by a key that is kept secret from the foundry. They {\em activate} the IC functionality by installing the key into a tamper-proof memory after fabrication.

\subsection{Related Work}\label{sec:related}

Logic locking is a popular technique to protect the intellectual property of ICs~\cite{llcarxiv}. Designers can apply logic locking at different abstraction levels and configure the protection based on the information available to the attacker~\cite{10.1145/3342099}.

Many existing methods operate on the gate-level netlists~\cite{llcarxiv}. Gate-level locking can not obfuscate all the semantic information because logic synthesis and optimizations absorb much of it into the netlist before the locking step. For example, constant propagation absorbs and propagate the constants. Our method completely strips the constants from the circuit before synthesis. Recently, alternative high-level locking methods obfuscate the semantic information before logic optimizations embed them into the netlist~\cite{dac18,rtlOzgur}. For example, TAO applies obfuscations during HLS~\cite{dac18} but requires access to the HLS source code to integrate the obfuscations and cannot obfuscate existing IPs. Protecting a design at the register-transfer level (RTL) is an interesting compromise~\cite{obfuscateDSPRTL,5401214}. Most of the semantic information (e.g., constants, operations and control flows) is still present in the RTL and obfuscations can be applied to existing RTL IPs. In~\cite{obfuscateDSPRTL}, the authors propose structural and functional obfuscation for DSP circuits. We propose a more general method that can be applied to any type of circuit. In~\cite{5401214}, the authors propose a method to insert a special finite state machine to control the transition between obfuscated mode (incorrect function) and normal mode (correct function). Such transition can only happen with a specific input sequence. We use a similar method to obfuscate the operations without additional logic (and power-up overhead) to make the circuit functional in normal mode. To obfuscate the semantic information, ASSURE leverages prior work on software program obfuscation~\cite{collberg91,BEHERA2015757,xu2017}. These methods obfuscate data structures, control flows, and constants through code transformations or by loading information from memory at runtime. We use a similar approach to create {\em opaque predicates} dependent only on the locking key~\cite{5401214}.

When the attackers have access only to the circuit netlist (like in the early stages of the fabrication process), they need to identify the correct variant among the ones created by the locking key. Redundancy attacks can recover part of the key bits for several locking methods~\cite{8714955}, while machine learning can predict the key bit values based on the structure of the circuit~\cite{sisejkovic2020challenging}. However, such attacks cannot recover what is not present in the circuit, like extracted constants, and cannot distinguish semantically equivalent variants. 
When the attackers can access also an activated IC (i.e. the oracle)~\cite{metrics2018,Shen:2019:BBS:3287624.3287670}, they can use Boolean Satisfiability (SAT)-based attacks to recover the key. Several solutions have been proposed to thwart SAT-based attacks~\cite{8279462,Yasin:2017:PLL:3133956.3133985}. For example, stripped-functionality logic locking (SFLL) extracts part of the functionality, which is hidden and restored upon the application of the correct key~\cite{Yasin:2017:PLL:3133956.3133985}. SFLL-HLS is the corresponding HLS-level extension~\cite{8942150}. However, complete protection is not guaranteed as attacks on SFLL have been reported when the ``protected'' functional inputs are at a certain Hamming distance from the key~\cite{sfll19,Deepak19}. Also, many SAT-resilient protections, like SARLock and SFLL, can be broken even with oracle-less attacks~\cite{8714955}. So, the approaches for the different threat models are complementary and must be combined to obtain multi-level protection. 

In this work, we aim at avoiding the attackers can recover the circuit with modern oracle-less attacks. {\em\uline{We base our techniques on the concept of {\bf indistinguishability}: all Boolean functions generated by a locking key have the same probability of being the correct circuit}}. 
So, netlist-only attacks are not able to identify and rule out ``incorrect'' designs.

\subsection{Paper Contributions}

ASSURE RTL obfuscation uses three techniques to obfuscate constants, arithmetic operations, and control branches. ASSURE provides the following contributions with respect to the state-of-the-art approaches: 
\begin{itemize}[leftmargin=*,topsep=0pt,noitemsep]
\item the three ASSURE techniques are complete and provably secure for creating {\em indistinguishable RTL designs} with no limitations on the input descriptions to be protected;
\item ASSURE can provide {\em multi-level security} together with oracle-based protections (e.g., scan-chain isolation~\cite{9214869});
\item ASSURE is a {\em technology-independent tool} that is fully compatible with existing EDA design flows and leaves complete control to the designer on the obfuscation process.
\end{itemize}
We describe our RTL-to-RTL translation framework in Section~\ref{sec:assure}, along with security proofs of obfuscations for constants, operations, and branches (Section~\ref{sec:techniques}). We also assess security against state-of-the-art oracle-less attacks (Section~\ref{sec:security}) and evaluate the related overhead (Section~\ref{sec:cost}).

\section{Threat Model: Untrusted Foundry}\label{sec:threat_model}
The state-of-art in logic locking considers two broad categories of threat models: netlist-only and oracle-guided~\cite{10.1145/3342099,8741035}. In both settings, the attacker has access to a locked netlist, but in the latter, also to an unlocked IC (\textit{oracle}) to analyze input/output relationships. 
The netlist-only model applies for an untrusted foundry that accesses the IC design for the first time. It also captures low-volume settings -- e.g., the design of future defense systems with unique hardware requirements~\cite{security} -- where the attacker would not reasonably be able to access a working copy of the IC. 
Consider, for instance, a fab-less defense contractor that outsources the fabrication of an IC to an {\bf untrusted foundry}. 
The untrusted foundry has access to the layout files of the design and can reverse engineer a netlist and even extract the corresponding RTL~\cite{7100906}. However, since the foundry produces the first ever batch of an IC design (in some cases the only one), an activated chip is not available through any other means. Attacks that rely on knowledge of an IC's true I/O behaviour (e.g., SAT attacks) cannot be applied and are therefore out-of-scope. However, the attacker can still rely on a range of netlist-only attacks, \textit{desynthesis}~\cite{DBLP:journals/corr/MassadZGT17}, \textit{redundancy identification}~\cite{8714955} and \textit{ML-guided structural and functional analysis}~\cite{sail,surf}, for instance, to recover the key bits and reverse engineer the locked netlist. In the following, we prove the resilience of ASSURE obfuscation to not only these three attacks, but also that ASSURE locked netlists reveal \emph{no information} about the design other than any prior knowledge that the designer might have about the design.
In the oracle-guided model, the attackers need to get an unlocked IC from the market -- e.g., because of high-volume commercial fabrication -- to analyze I/O relationships and apply the corresponding attacks. With our method, we thwart attacks that are successful for oracle-guided protections even without activated IC.

\begin{figure}
\centering
\includegraphics[width=\columnwidth]{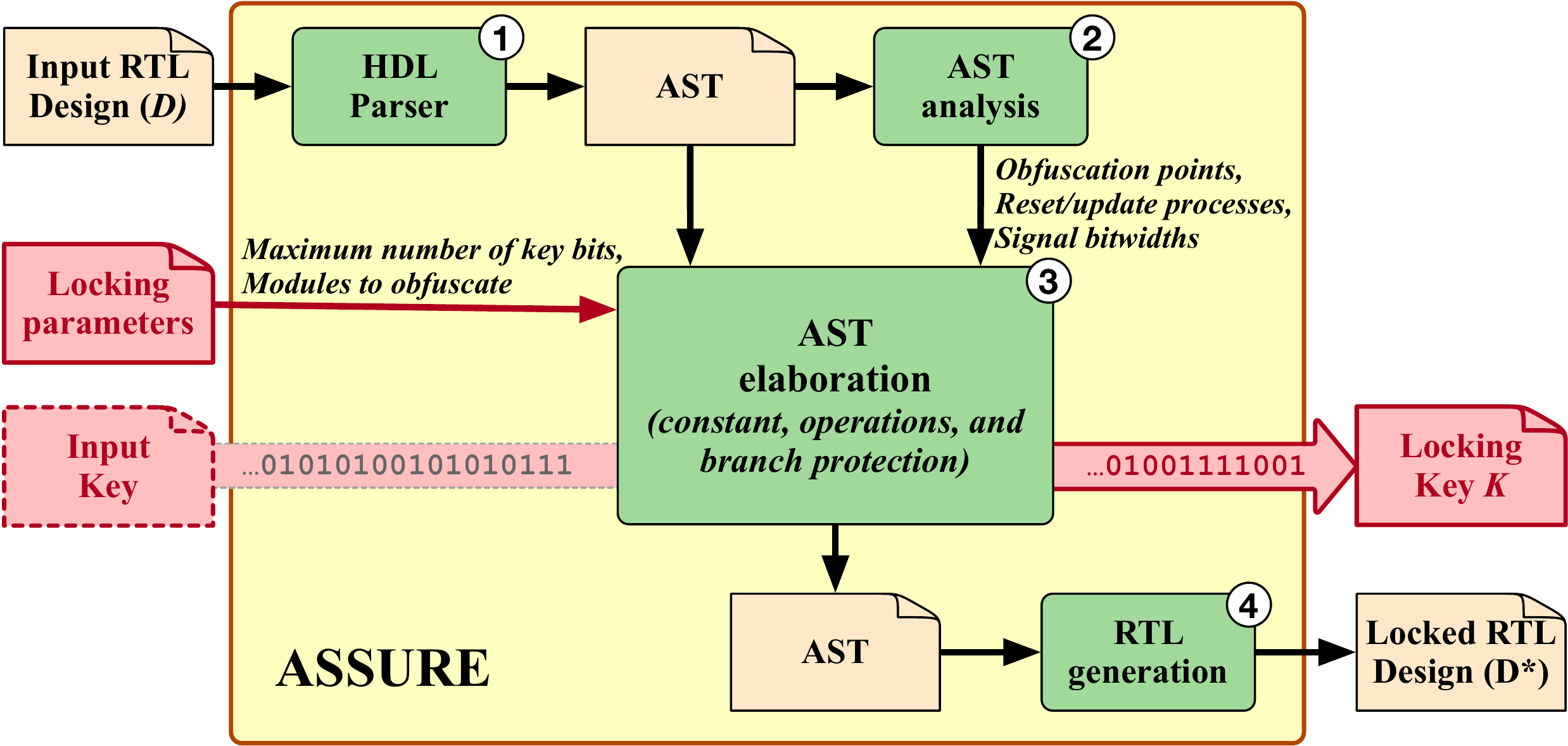}
\caption{Organization of ASSURE.}\label{fig:assure}
\end{figure}

\section{Overview of ASSURE}\label{sec:assure}
Fig.~\ref{fig:assure} shows the ASSURE flow. Given an RTL design~$D$ and a set of obfuscation parameters, ASSURE generates a design $D^*$ together with a single locking key $\mathcal{K}^{*}_r$ such that $D^*$ matches the functionality of $D$ only when $\mathcal{K}^{*}_r$ is applied. 
ASSURE is a technology-independent and operates on the RTL after system integration but before logic synthesis. ASSURE obfuscates existing IPs and those generated with commercial HLS tools. Even if logic locking is a hardware approach, obfuscating RTL code has analogies with {\em program obfuscation} to protect the software IP~\cite{collberg91,xu2017}. In both cases, the designer aims at obfuscating the semantic information contained into the design rather than its structure~\cite{5401214}.
ASSURE obfuscates the RTL by adding in {\em opaque predicates} such that the  evaluation of the opaque predicates depends on the locking key; their values are known to the designer during obfuscation, but unknown to the untrusted foundry.  ASSURE obfuscates three {\bfseries semantic elements} useful to replicate function of an IP:
\begin{itemize}[leftmargin=*,topsep=0pt,noitemsep]
\item {\em constants} contain sensitive information in the computation (e.g., filter coefficients).
\item {\em operations} determine functionality.
\item {\em branches} define the execution flow (i.e., which operations are executed under specific conditions).
\end{itemize}
ASSURE parses the input HDL and creates the {\em abstract syntax tree} (AST) -- step \circled{1}. It then analyzes the AST to select the semantic elements to lock (step~\circled{2}) and obfuscates them ({\em AST elaboration} -- step~\circled{3}). The {\em RTL generation} phase (step~\circled{4}) produces the output RTL design that has the same external interface as the original module, except for an additional input port that is connected to the place where $\mathcal{K}^{*}_r$ is stored. 
ASSURE starts from a synthesizable IP and  modifies its description, it fits with existing EDA flows and same constraints as the original, including tools to verify that resulting RTL is equivalent to the original design when the correct key is used and to verify that it is not equivalent to the  original when an incorrect key is used. 

The key idea of ASSURE is that the functionality of $D^*$ is much harder to understand without the parameter~$\mathcal{K}^{*}_r$. If the attackers apply a  key different from $\mathcal{K}^{*}_r$ to $D^*$, they obtain plausible but wrong circuits, indistinguishable from the correct one. These variants are {\it indistinguishable} from one another without  a-priori knowledge of the design. 

\subsection{ASSURE Obfuscation Flow}\label{sec:assure_analysis}
To generate an obfuscated RTL design, we must match the requirements of the IP design with the constraints of the technology for storing the key (e.g., maximum size of the tamper-proof memory). On one hand, the number of bits needed to obfuscate the semantics of an RTL design depends on the complexity of the algorithm to protect. On the other hand, the maximum number of key bits that can be used by ASSURE ($K_{max}$) is a design constraint that depends on the technology for storing them in the circuit. ASSURE analyzes the input design to identify which modules and which circuit elements in modules must be protected. First, ASSURE does {\bf depth-first analysis} of the design to {\em uniquify} the module hierarchy and creates a list of modules to process. In this way, {\em\uline{ASSURE hides the semantics of the different modules so that extracting knowledge from one instance does not necessarily leak information on all modules of the same type}}.

\begin{algorithm}[t!]
	\footnotesize
	\SetKwFunction{ObfuscateModule}{ObfuscateModule}
	\SetKwFunction{CreateBlackList}{CreateBlackList}
	\SetKwFunction{DepthFirstAST}{DepthFirstAST}
	\SetKwFunction{BitReq}{BitReq}
	\SetKwFunction{Obfuscate}{Obfuscate}
	\SetKwFunction{GetObfuscationKey}{GetObfuscationKey}
	\SetKwFunction{KeyLength}{KeyLength}
	\SetKwProg{obfuscateModule}{Procedure}{}{}
	
	\caption{ASSURE obfuscation.\label{alg:obfuscate}}
	\obfuscateModule{\ObfuscateModule{$AST_m$, $\mathcal{K}^*_r$, $K_{max}$}}{
		\KwData{$AST_m$ is the AST of the module $m$ to obfuscate}
		\KwData{$\mathcal{K}^*_r$ is the current locking key}
		\KwData{$K_{max}$ is the maximum number of key bits to use}
		\KwResult{$AST^*_m$ is the obfuscated AST of the module $m$}
		\KwResult{$\mathcal{K}^*_r$ is the updated locking key}		
		
		$\mathtt{BlackList} \gets $ \CreateBlackList{$AST_m$}\;
		$AST^*_m \gets \mathtt{BlackList}$\;
		$\mathtt{ObfElem} \gets $ \DepthFirstAST{$AST_m$}~$\setminus$~$\mathtt{BlackList}$\;
		\ForEach{$el \in \mathtt{ObfElem}$}{
		   $b_{el} \gets $ \BitReq{$el$}\;
		   \eIf {\KeyLength{$\mathcal{K}^*_r$}$ + b_{el} > K_{max}$}{
		      $AST^*_m \gets AST^*_m \cup el$\;
		   } {
		      $K_{el} \gets $ \GetObfuscationKey{$el$}\;
		      $AST^*_m \gets AST^*_m \cup $ \Obfuscate{$el$, $K_{el}$}\;
		      $\mathcal{K}^*_r \gets \mathcal{K}^*_r \cup K_{el}$\;
           }
	}
	\Return{$\{AST^*_m,\mathcal{K}^*_r\}$}
}
\end{algorithm}

After uniquifying the design, ASSURE analyzes the AST of each module with Algorithm~\ref{alg:obfuscate} starting from the innermost ones. Given a hardware module, ASSURE first creates a ``black list'' of the elements that must be excluded from obfuscation (line 2). For example, the black list contains elements inside reset and update processes or loop induction variables (see Section~\ref{sec:techniques}). The designer can also annotate the code to specify that specific regions or modules must be excluded from obfuscation (e.g., I/O processes or publicly-available IPs). The black-list elements are added unchanged to the output AST (line 3). Finally, ASSURE determines the list of AST elements to obfuscate (line 4) and process them (lines 5-12). The resulting list \texttt{ObfElem} follows the visit order of the depth-first search. For each element, ASSURE computes the number of bits required for obfuscation (line 6) and check if there are enough remaining key bits (line 7). If not, ASSURE does not obfuscate the element (line 8). Indeed, reusing a key bit across multiple elements as in~\cite{dac18} reduces the security strength of our scheme because extracting the key value for one element invalidates the obfuscation of all others sharing the same key bit. If the obfuscation is possible (lines 9-12), ASSURE generates the corresponding key bits to be added to $\mathcal{K}^*_r$ (line 10). These bits depend on the specific obfuscation technique to be applied to the element and can be randomly generated, extracted from an input key (see Fig.~\ref{fig:assure}), or extracted from the element itself (see Section~\ref{sec:techniques}). ASSURE uses these key bits to obfuscate the element and the result is added to the output AST (line 11). The key bits are also added to the output locking key (line 12). We repeat this procedure for all modules until the top, which will return the AST of the entire design and the final key. This approach leaves full control to the designers that can explore trade-offs by providing constraints on the number of bits and combining the depth-first analysis with the annotations to exclude elements from obfuscation.

\subsection{ASSURE Obfuscations and Security Proofs}\label{sec:techniques}

Each of the ASSURE techniques targets an essential element to protect and uses a distinct part of the $r$-bit locking key $\mathcal{K}^{*}_r$, to create an opaque predicate\footnote{We use Verilog notation in the examples, but the approach is general.}. In software, an {\em opaque predicate} is a predicate for which the outcome is certainly known by the programmer, but requires an evaluation at run time~\cite{collberg91}. We create {\bf hardware opaque predicates}, for which the outcome is determined by ASSURE (and so known) at design time, but requires to provide the correct key at run time. Any predicate involving the extra parameter $\mathcal{K}^{*}_r$ meets this requirement. 
Given a locking key $\mathcal{K}^{*}_r$, {\em\uline{ASSURE generates a circuit \textbf{indistinguishable} from the ones  generated with any other $\mathcal{K}_r\neq \mathcal{K}^{*}_r$ when the attacker has no prior information on the design}}. 

We show that ASSURE techniques offer provable security guarantees~\cite{DBLP:journals/corr/MassadZGT17}. 
Consider an $m$-input $n$-output Boolean function 
$\mathcal{F}$ : $X \rightarrow Y$, where $X \in \{0,1\}^{m}$ and $Y \in \{0,1\}^{n}$. 

{\bf\noindent Definition} Obfuscation $\mathcal{O}(\mathcal{K}^{*}_r)$ transforms $\mathcal{F}$ into an $m+r$-input $n$-output function $\mathcal{L}_{\mathcal{K}^*_r}$ defined as:
\begin{equation}
    \mathcal{L}_{\mathcal{K}^*_r}(X,K) = \mathcal{O}(\mathcal{F}(X),\mathcal{K}^*_r)
\end{equation}
where 
$\mathcal{L}_{\mathcal{K}^*_r}$: $X \bigtimes K \rightarrow Y$ and $K \in \{0,1\}^{r}$ are such that
\begin{itemize}
\leftskip =-10pt
\rightskip =-10pt
    \item $\mathcal{L}_{\mathcal{K}^*_r} (X,\mathcal{K}^{*}_r) = \mathcal{F}_{\mathcal{K}^*_r}(X) = \mathcal{F}(X)$
    \item $\mathcal{L}_{\mathcal{K}^*_r} (X,\mathcal{K}_r) = \mathcal{F}_{\mathcal{K}_r}(X) \neq$ $\mathcal{F}(X)$ when $\mathcal{K}_r\neq \mathcal{K}^{*}_r$
\end{itemize}

This definition shows $\mathcal{C}_{\mathcal{K}^*_r}$ represents a family of Boolean functions $\{\mathcal{F}_{K_r}\}$ based on the generic $r$-bit key input $\mathcal{K}_r$. The functionality $\mathcal{F}(X)$ can be re-obtained uniquely with the correct key $\mathcal{K}^*_r$. This is followed by a corollary about a characteristic of the family of Boolean functions  $\mathcal{L}_{\mathcal{K}^*_r}$.
\begin{theorem} 
For an obfuscated netlist $L_{\mathcal{K}^*_r}(X,K)$ created from $\mathcal{F}(X)$ with obfuscation $\mathcal{O}(\mathcal{K}^{*}_r)$, the unlocked functions $\mathcal{F}_{K_1}$ and $\mathcal{F}_{K_2}$ (for keys $K_1$ and $K_2$) relate as follows:
\begin{gather}
    \mathcal{F}_{K_1} \neq \mathcal{F}_{K_2} \forall K_1, K_2 \in K: K_1 \neq K_2
\end{gather}
\end{theorem}
\begin{proof}
Let us first consider case (i) $K_1 = \mathcal{K}^*_r$. Therefore, by the definition of RTL obfuscation scheme $\mathcal{O}$, $\mathcal{F}_{K_1} \neq \mathcal{F}_{K_2} \forall K_2 \in K, K_1 \neq K_2$. Now, for case (ii) $K_1 \neq  \mathcal{K}^*_r$, there exists a locked netlist $\mathcal{L}_{\mathcal{K}^*_r}$ that locked $\mathcal{F}_{K_1}$ using $K_1$. Therefore, $\mathcal{F}_{K_2}$ = $\mathcal{L}_{\mathcal{K}^*_r}(X,\mathcal{K}_2)$. By the definition of logic locking security, $\mathcal{F}_{K_2} \neq \mathcal{F}_{K_1} \forall K_2 \neq K_1$ in $\mathcal{L}_{\mathcal{K}^*_r}(X,\mathcal{K}_r)$.
\end{proof}

\begin{figure*}
     \centering
     \subfloat[][Constant obfuscation]{\includegraphics[height=2.5cm]{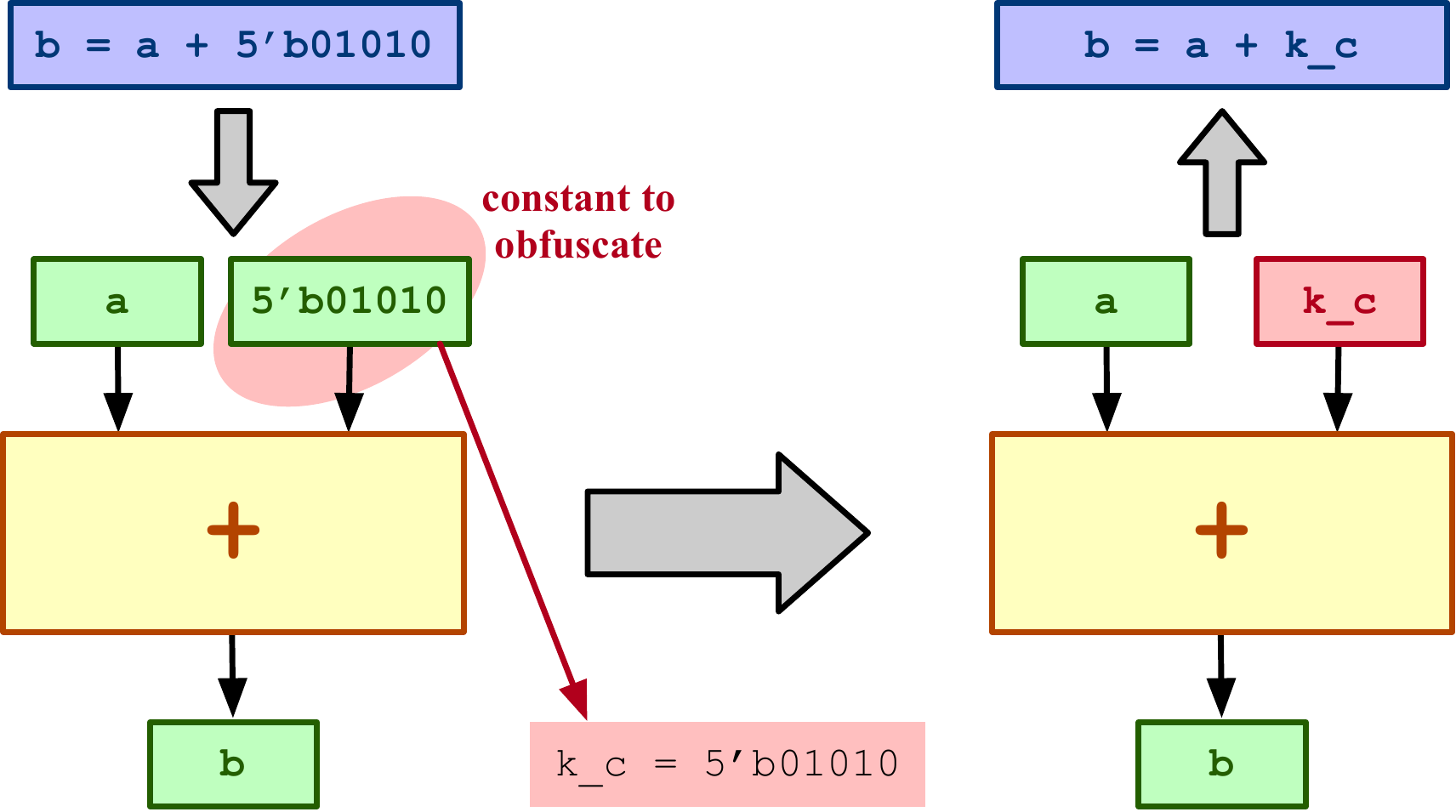}\label{fig:const_obfuscation}}
     \hfill
     \subfloat[][Operation obfuscation]{\includegraphics[height=2.5cm]{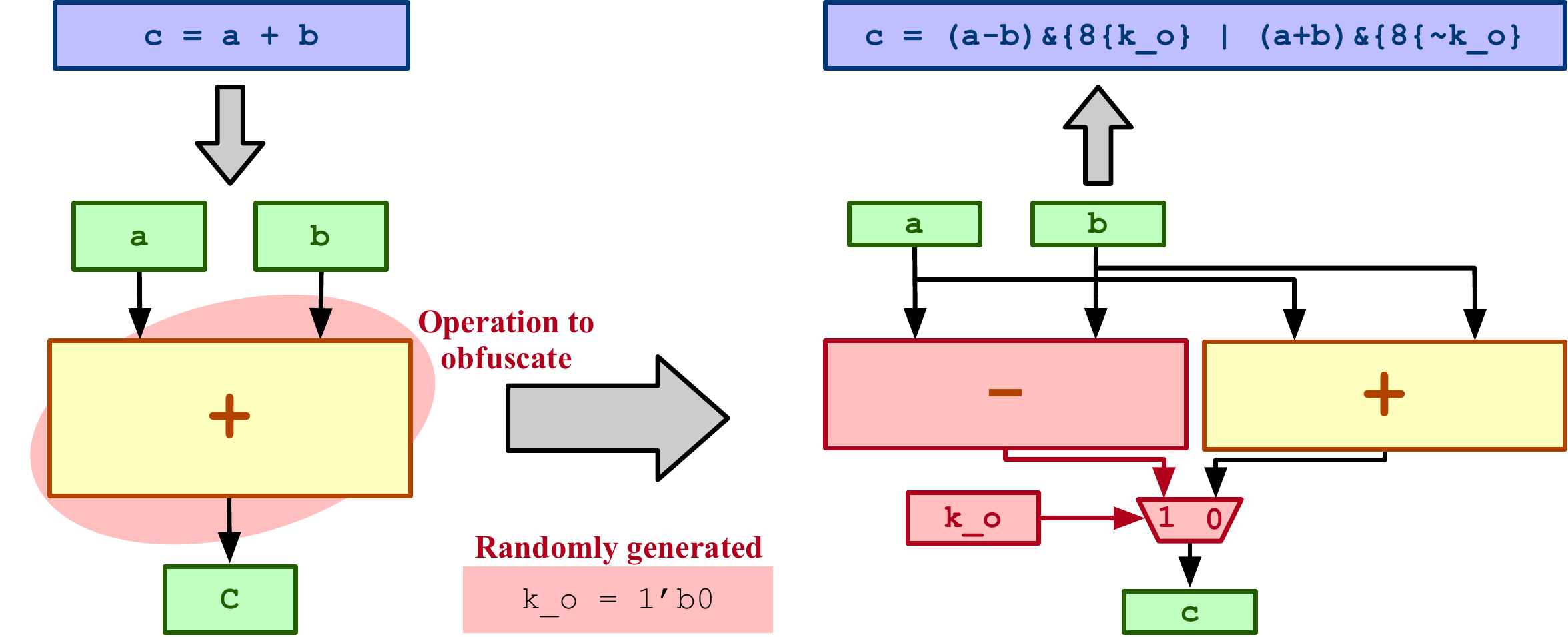}\label{fig:op_obfuscation}}
     \hfill
     \subfloat[][Branch obfuscation]{\includegraphics[height=2.5cm]{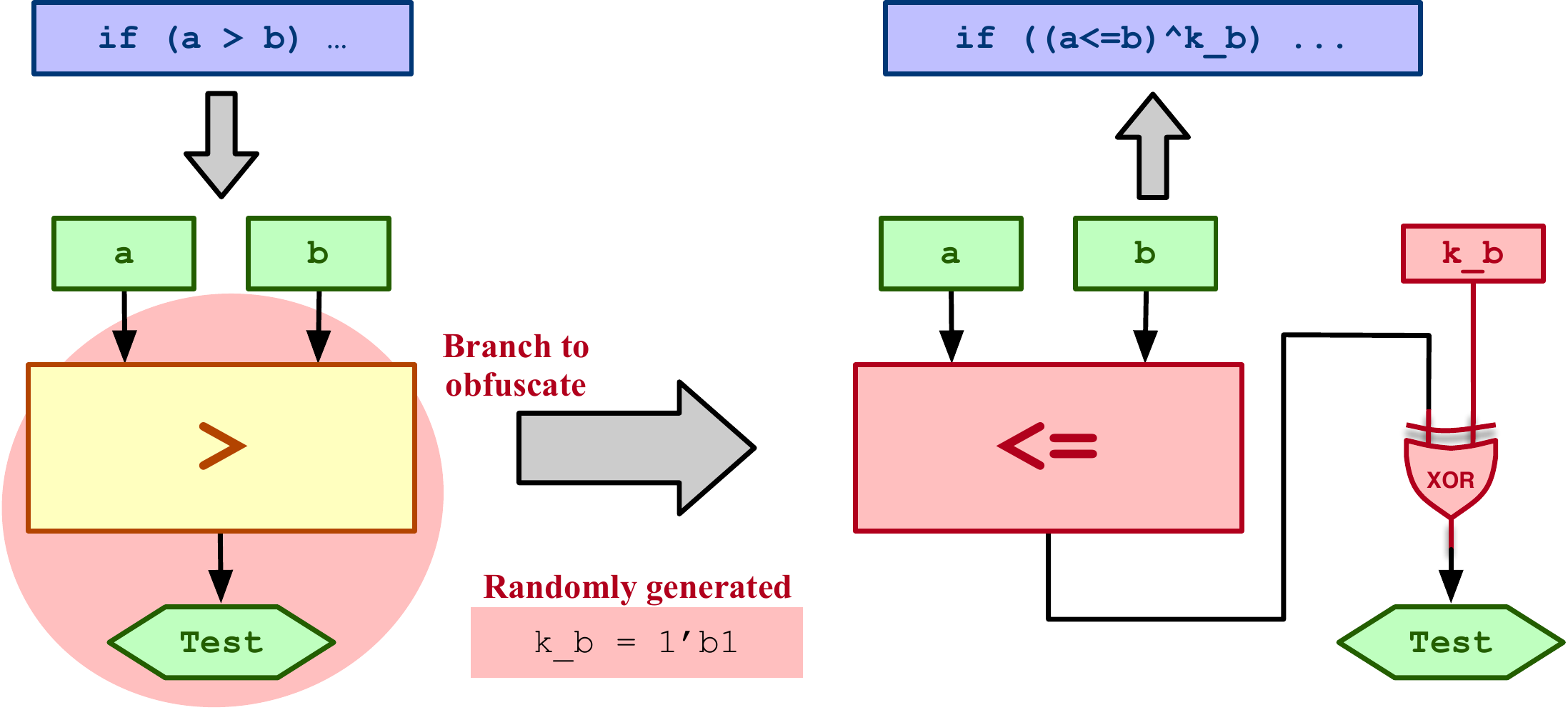}\label{fig:branch_obfuscation}}
     \caption{Three ASSURE obfuscations. (a) Constant, (b) Operation, and (c) Branch. Each obfuscation uses a portion of the key.}
     \label{fig:techniques}
\end{figure*}

We define $P[\mathcal{L}_{\mathcal{K}_r} | \mathcal{O}(\mathcal{F}(X),\mathcal{K}]$ as the probability of obtaining the locked design $\mathcal{L}_{\mathcal{K}_r}$ given that we locked the Boolean function $\mathcal{F}(X)$ applying $\mathcal{O}$ with $\mathcal{K}$. The RTL locking scheme $\mathcal{O}$ is secure under the netlist-only threat model:
\begin{theorem}\label{theo:strong}
A logic locking scheme $\mathcal{O}$ for $r$-bit key $K$ is secure for a family of Boolean functions $\mathcal{F}_{K_r}$ of cardinality $2^r$ if the following condition holds true:
\begin{gather}     
P[\mathcal{L}_{\mathcal{K}'_r} | \mathcal{O}(\mathcal{F}(X),\mathcal{K}^*_r)] = P[\mathcal{L}_{\mathcal{K}'_r} |\mathcal{O}(\mathcal{F}_{K_r}(X),\mathcal{K}_r)]\nonumber \\ \forall \mathcal{K}_r \neq \mathcal{K}^*_r, \mathcal{F}(X) \neq \mathcal{F}_{K_r}(X)
\end{gather}
\end{theorem}
\noindent This theorem states any locking key $\mathcal{F}_{K_r}$ is equally probable to generate the locked netlist $\mathcal{L}_{\mathcal{K}'_r}$ generated by the locking scheme $\mathcal{O}$, creating a family of Boolean function $\mathcal{F}_{K_r}(X)$ all having the same probability to be the original Boolean function $\mathcal{F}(X)$. We show all our obfuscations satisfy these two claims, providing a security guarantee of $2^r$ under the proposed threat model. This guarantee allows the designer to choose the parameter $r$  to match the technology issues for storing the bits in the final IC. ASSURE will generate a locked design with the corresponding level of security.

\subsubsection{Constant Obfuscation} This obfuscation removes selected constants and moves them into the locking key $K$ as shown in Fig.~\ref{fig:const_obfuscation}. The original function is preserved only when the key provides the correct constant values. {\em\uline{Each constant bit is a {\em hardware-opaque predicate};  the designer knows its value and the circuit operation depends on it}}. 

\noindent
{\em {\bf Example: } Consider the RTL operation {\tt b = a + 5'b01010}. To obfuscate the constant, we add a 5 bit key {\tt K\_c = 5'b01010}. The RTL is rewritten as {\tt b = a + K\_c}.
 The attacker has no extra information and $2^{5}$ possibilities from which to guess the correct value.\hfill$\Box$}

Hiding constant values allows designers to protect sensitive information (e.g., proprietary implementations of digital filters or cryptographic algorithms~\cite{10.1007/978-3-662-43826-8_10}) but also may prevent subsequent logic optimizations (e.g., constant propagation and wire trimming). However, several constants are unuseful and, in some cases, problematic to protect. For example, reset values are set at the beginning of the computation to a value that is usually zero and then assigned with algorithm-related values. 
Also, obfuscating reset polarity or clock sensitivity edges of the processes introduces two problems: incorrect register inferencing, which leads to synthesis issues of the obfuscated designs, and incorrect reset process that easily leads to identify the correct key value. In particular, if we apply obfuscation to the reset processes and the attacker provides an incorrect key value, the IC will be stalling in the reset state when it is supposed to be in normal execution. 
So, we exclude constants related to reset processes and sensitivity values from obfuscation.

{\em\uline{Proof}}. The structure of the obfuscated circuit is independent of the constant and, given an $r$-bit constant, the $2^r$ values are indistinguishable. The attacker cannot get insights on the constants from the circuit structure. ASSURE constant obfuscation satisfies the provable security criteria of logic locking $\mathcal{L}$ under strong adversarial model as defined in Theorem~\ref{theo:strong}. 

Let us consider an RTL design of $m$ inputs and $n$ outputs $R : X \rightarrow Y$, $X \in \{0,1\}^m$ and uses an $r$-bit constant $C^*_{r}$. ASSURE constant obfuscation converts the $r$-bit constant into an $r$-bit key $K^*_r$ as a lock $\mathcal{O}$ and uses it to lock the design $\mathcal{L}_{\mathcal{K}^*_r}$. The obfuscated constant $C_r$ is depicted as follows:
\begin{gather}
    C_r = \mathcal{K}_r
\end{gather}
where, $C_r = C^*_r$ only when $\mathcal{K}_r = \mathcal{K}^*_r = C^*_r$.

\begin{quote}
\leftskip =-20pt
\rightskip =-20pt
\textbf{Claim 1}: Any unlocked constant $C_{K_1}$ and $C_{K_2}$ using r-bit keys $K_1$ and $K_2$ are unique. (Theorem 1)
\end{quote}
\begin{proof}
 $\forall K_1 \neq K_2$, $K_1, K_2 \in \{0,1\}^r$ $\implies C_{K_1}\neq C_{K_2}$. 
\end{proof}

\begin{quote}
\leftskip =-20pt
\rightskip =-20pt
\textbf{Claim 2}: A constant-obfuscated circuit with r-bit key $K$ can be generated from $2^r$ possible constants (each of $r$-bit) with equal probability, i.e. the following holds true.
\begin{gather}     
P[C_r | K = \mathcal{K}^*_r] =  P[C_r | K = K_r] \nonumber\\
    \forall K_r \neq K^*_r ; K_r \in 2^r
\end{gather}
\end{quote}
\begin{proof} 
The probability of choosing $K_r$ is uniform. So, \\
$P[K = \mathcal{K}^*_r] = P[K = \mathcal{K}_r]$, $\forall K_r \neq K^*_r$ \\
$\implies P[C^*_r] = P[C_{r}]$, $C^*_r \neq C_r, \forall C_r \in \{0,1\}^r$.
\end{proof}

Claims 1 and 2 jointly denote that the constant obfuscated by $2^r$ unique constants are indistinguishable and can be unlocked uniquely by the correct $r$-bit key. Constant obfuscation hides the original constants with a security strength of $2^r$.

In Fig.~\ref{fig:constantExtraction}, we show area overhead of \texttt{DES3} and \texttt{RSA}, two CEP benchmarks~\cite{cep}. This experiment shows that constant obfuscation generates indistinguishable circuits. We consider a variable from each benchmark: \texttt{sel\_round} from \texttt{DES3} and \texttt{modulus\_m1\_len} from \texttt{RSA}. We generate different circuits by assigning different constants to the same variable. We synthesize these circuit variants and obtain the area overhead. Fig.~\ref{fig:constantExtraction} shows that every constant value ($c1-c5$) can be reverse engineered from the synthesized circuit since each constant directly maps to unique area overhead. On the contrary, the area overhead of synthesized circuits remain the same after obfuscation, and the obfuscated circuits are indistinguishable, making it  difficult for the attacker to recover the constant.

\subsubsection{Operation Obfuscation} We generate a random key bit and use it to multiplex the operation result with that from another operation sharing the same inputs, as shown in Fig.~\ref{fig:op_obfuscation}. {\em\uline{The mux selector is a {\em hardware opaque predicate} because the designer knows its value and the mux propagates the correct result only for the correct key bit}}. This is similar to that proposed for C- and HLS-level obfuscation~\cite{dac18,8715083}.

\noindent{\em {\bf Example:}  Let us obfuscate RTL operation {\tt c = a + b} with a dummy subtraction. We generate a key bit {\tt k\_o = 1'b0} and rewrite the RTL as {\tt c = k\_o ? a - b : a + b}. The original function is selected for the correct  {\tt k\_o}.}\hfill$\Box$

The ternary operator is a simple representation of the multiplexer, but it may impact code coverage. It introduces extra branches in the circuit, where one of the paths is never activated once the key is provided. To keep the same coverage as the original design, we rewrite the mux selection as  {\tt o = in1 \& k | in2 \& $\mathtt{\sim}$k}.

\noindent
{\em {\bf Example:}  Operation {\tt c = a + b} obfuscated as {\tt c = k\_o ? a - b : a + b} can be written as {\tt c = (a - b)\&\{8\{k\_o\}\} | (a + b)\&\{8\{$\mathtt{\sim}$k\_o}\}\}. This is equivalent to ternary operation without branches, and has the same code coverage.}\hfill$\Box$

Since operations use the same inputs, ASSURE adds a multiplexer at the output with  its select connected  to the key bits. The multiplexer and the additional operator are area overhead. The multiplexer impacts the critical path and the additional operation introduces a delay when it takes more time than the original one.  
We create a pool of alternatives for each operation type. Original and dummy operations are ``balanced'' in complexity to avoid increasing the area and the critical path.
Dummy operations are selected to avoid structures the attacker can easily unlock. Incrementing a signal by one cannot be obfuscated by a multiplication by one, clearly a fake. Dummy operators are also selected to avoid collisions. For example, adding a constant to a signal cannot be obfuscated with a subtract because the wrong operation key bit can activate the circuit when the attacker provides the two's complement of the  constant. 

{\em\uline{Proof}}. 
Consider an RTL design with $m$ inputs and $n$ outputs, with a mapping $\mathcal{F} : X \rightarrow Y$, $X \in \{0,1\}^m$ and with $r$ possible sites for operator obfuscation. ASSURE uses multiplexer-based locking $\mathcal{O}$ with an $r$-bit key $\mathcal{K}^*_r$ to lock the design and generate $\mathcal{L}_{\mathcal{K}^*_r}$. 
\begin{gather}
   \mathcal{L}_{\mathcal{K}^*_r} = \mathcal{F}(X,k_1,k_2,..,k_r) \nonumber\\  = \overline{k_1}\mathcal{F}(X,0,k_2,..,k_r) + k_1 \mathcal{F}(X,1,k_2,..,k_r) \nonumber\\
   = {\mathcal{K}^1_r}\underbrace{\mathcal{F}(X,K=\mathcal{K}^1_r)}_{\mathcal{F}_{K_1}} +  {\mathcal{K}^2_r}\underbrace{\mathcal{F}(X,K=\mathcal{K}^2_r)}_{\mathcal{F}_{K_2}} + .. \nonumber\\
   .. + \mathcal{K}^{2^r}_r\underbrace{\mathcal{F}(X,K = \mathcal{K}^{2^r}_r)}_{\mathcal{F}_{K_{2^r}}}
\end{gather}
where, $\mathcal{F}_{K^*_r}(X) = \mathcal{L}_{\mathcal{K}^*_r}(X,K=\mathcal{K}^*_r)$ ($\mathcal{K}^*_r$ is the $r$-bit key). Each location of operator obfuscation applies output of different operations (one original and another fake) to a multiplexer. The following equation holds true for operator obfuscation.
\begin{gather}
   \mathcal{F}(X,k_1,..,k_i=0,..,k_r) \neq \mathcal{F}(X,k_1,..,k_i=1,..,k_r) \nonumber\\
   \forall i \in [1,r] \label{eq7}
\end{gather}
Secondly, the sites of operation obfuscation are different. The output of multiplexer using any key-bit value at one location is independent of the choice made elsewhere. Given a key $K$, the unlocked function of two circuits will be different if we set same logic value at two different key-bit locations. For an example $K=1101$, if one chooses bit location 2 and 4 and flip them, i.e. $K_1 = 1001, K_2=1100$, then $F_{K_1} \neq F_{K_2}$.
\begin{gather}
   \mathcal{F}(X,k_1,..,k_i=\overline{k_i},..,k_r) \neq \mathcal{F}(X,k_1,..,k_j=\overline{k_j},..,k_r) \nonumber \\ \forall  i,j \in [1,r], i \neq j \label{eq8}
\end{gather}

\begin{quote}
\leftskip =-20pt
\rightskip =-20pt
\textbf{Claim 1}: Any pair of unlocked circuit $F_{K^1_r}$ and $F_{K^2_r}$ using r-bit keys $K^1_r$ and $K^2_r$ on multiplexer-based obfuscated circuit $\mathcal{L}_{\mathcal{K}^*_r}$ are unique. (Theorem 1)
\end{quote}
\begin{proof}
 $\forall K^1_r \neq K^2_r$, $K^1_r, K^2_r \in \{0,1\}^r$ \\
 $\implies$ Hamming distance $(K_1,K_2) \in [1,r]$. \\
 $\implies$ Eq.~\ref{eq7} + Eq.~\ref{eq8}, $F_{K_1} \neq F_{K_2}$
\end{proof}
\begin{quote}
\leftskip =-20pt
\rightskip =-20pt
\textbf{Claim 2}: MUX-based obfuscation with r-bit key $K$ can be generated from $r$ different locations having $2^r$ operations with equal probability, i.e. following condition holds true.
\begin{gather*}     
P[\mathcal{L}_{\mathcal{K}_r} | \mathcal{O}(\mathcal{F}(X), K^*_r)] = P[\mathcal{L}_{\mathcal{K}_r} | \mathcal{O}(\mathcal{F}(X), K^i_r)] \\
    \forall K^i_r \neq K^*_r ; F_{K^i_r} \neq F_{K^*_r}; i \in [1,2^r]
\end{gather*}
\end{quote}

\begin{proof} 
The probability of choosing $K_r$ is uniform. Therefore, \\
P[K = $K^*_r$] = P[K = $K^i_r$], $\forall K^i_r \neq K^*_r$ \\
$\implies P[\mathcal{L}_{\mathcal{K}_r}(X, K = \mathcal{K}^*_r]) = P[\mathcal{L}_{\mathcal{K}_r}(X,K = \mathcal{K}^i_r)]$ \\
$\implies P[\mathcal{F}_{\mathcal{K}^*_r}] = P[\mathcal{F}_{\mathcal{K}^i_r}] = \frac{1}{2^r}$.
\end{proof}

Claims 1 and 2 show that operator obfuscation can generate indistinguishable netlists.

In Fig.~\ref{fig:opObfuscation_eg}, we demonstrate area overhead of the two benchmark circuits \texttt{DES3} and \texttt{RSA} for operator obfuscation supporting our claims generate indistinguishable circuits. We consider a single operation from each benchmark: addition of \textit{auxiliary input} and \texttt{round\_output} from \texttt{DES3}, and subtraction of \texttt{modulus\_m1\_len} from a constant value in \texttt{RSA}. We generate different circuits by replacing the original operators with other operators. After synthesis, area overhead of these variants (Fig.~\ref{fig:opObfuscation_eg}) are unique and can be reverse engineered. On the contrary, the area overhead of synthesized circuits remain the same after obfuscation and so the obfuscated circuits reveals nothing about the original operator. 

\subsubsection{Branch Obfuscation}\label{sec:branch_obf}
To hide which branch is taken after the evaluation of an RTL condition, we obfuscate the test with a key bit as {\tt cond\_res xor k\_b}, as shown in Fig.~\ref{fig:branch_obfuscation}. To maintain semantic equivalence, we negate the condition to reproduce the correct control flow when {\tt k\_b = 1'b1} because the XOR gate inverts the value of {\tt cond\_res}. We apply De Morgan's law to propagate the negation to disguise the identification of the correct condition.
{\em\uline{The resulting predicate is {\em hardware-opaque} because the designer knows which branch is taken but this is unknown without the correct key bit}}. 

{\noindent\em {\bf Example:}  Let {\tt a > b} be the RTL condition to obfuscate with key {\tt k\_b = 1'b1}. We rewrite the condition as {\tt(a <= b)$\wedge$k\_b}, which is equivalent to the original one only for the correct key bit. The attacker has no additional information to infer if the original condition is {\tt >} or {\tt <=}.}\hfill$\Box$

Obfuscating a branch introduces a 1-bit XOR gate, so the area and delay effects are minimal. Similar to constant obfuscation, branch obfuscation is applied only when relevant. For example, we do not obfuscate reset and update processes. We apply the same technique  to  ternary operators. When these operators are RTL multiplexers, this technique thwarts the data propagation between the inputs and the output. The multiplexer propagates the correct value  with the correct key.

{\em\uline{Proof}}. For an $m$ input RTL design, we have a control-flow graph (CFG) $G(V,E)$ having $|V|$ nodes and $|E|$ edges. 
We do a depth-first-traversal of the CFG and order the $r$ conditional nodes in the way they are visited. Let the ordered set of conditional nodes be $V_{CN} = \{v_1,v_2,...v_r\}, V_{CN} \subset V$ ($r = |V_{CN}|$). ASSURE applies XOR-based branch obfuscation to $V_{CN}$ with $r$-bit key $\mathcal{K}^*_r$ as the logic locking scheme $\mathcal{O}$ and generates a locked design $\mathcal{C}_{\mathcal{K}^*_r}$. For example, if $V_{CN} = \{v_1,v_2,v_3,v_4\}$ and $K = 1101$, then $\mathcal{L}(V_{CN}) = \{\overline{v_1},\overline{v_2},v_3,\overline{v_4}\}$. The locked design, post branch-obfuscation is illustrated as follows.
\begin{gather}
    \mathcal{L}_{\mathcal{K}^*_r}(G(V,E),K)= \mathcal{O}(G(V,E),K^*_r) = \nonumber \\  
    G(\mathcal{O}(V_{CN},K^*_r)\cup (V \setminus V_{CN}),E) = \nonumber \\  
    G((V_{CN} \oplus K^*_r)\cup (V \setminus V_{CN}),E)
\end{gather}
where $G(V, E) = \mathcal{L}_{\mathcal{K}^*_r}(G(V,E),K = \mathcal{K}^*_r)$.

\begin{table*}[t]
\caption{Characteristics of the input RTL benchmarks.}\label{tab:benchmarks}
\scriptsize
\vspace{-6pt}
\begin{tabularx}{\linewidth}{@{}L{1.4cm}@{} L{1.5cm}@{} @{}R{1cm} || @{}R{0.9cm}@{} @{}R{0.9cm}@{} @{}R{1.2cm} | @{}R{1.1cm} || R{1.2cm}@{} R{1.2cm}@{} R{1.2cm}@{} R{1.1cm}@{} R{1.1cm} | R{1.4cm}@{} R{1.4cm}@{}}
\toprule
{\bf Suite} & {\bf Design} & {\bf Modules} & {\bf Const} & {\bf Ops} & {\bf Branches} & {\bf Tot Bits} &  {\bf Comb cells} & {\bf Seq cells} & {\bf Buf cells} & {\bf Inv cells} & {\bf \# nets} & {\bf Area ($\mu m^2$)} & {\bf Delay ($ns$)}\\
\midrule
\multirow{9}{*}{CEP}	&	{\tt AES}	&	657	&	102,403	&	429	&	1	&	819,726	&  127,667 & 8,502 & 506 & 21,812 & 136,493 & 42,854.69 & 136.75\\
	&	{\tt DES3}	&	11	&	4	&	3	&	775	&	898	&  2,076 & 135 & 128 & 368 & 2,448 & 736.96 & 192.28\\
	&	{\tt DFT}	&	211	&	447	&	151	&	132	&	8,697	&  118,201 & 38,521 & 9,552 & 41,320 & 158,807 & 81,865.94 & 336.72\\
	&	{\tt FIR}	&	5	&	10	&	24	&	0	&	344	&  820 & 439 & 49 & 225 & 1,704 & 1,129.36 & 377.76\\
	&	{\tt IDFT}	&	211	&	447	&	151	&	132	&	8,697	&  118,154 & 38,525 & 9,576 & 41,305 & 158,722 & 81,821.90 & 333.59\\
	&	{\tt IIR}	&	5	&	19	&	43	&	0	&	651	&  1,378 & 648 & 72 & 367 & 2,621 & 1,679.72 & 464.82\\
	&	{\tt MD5}	&	2	&	150	&	50	&	1	&	4,533	&  4,682 & 269 & 168 & 923 & 5,756 & 1,840.15 & 791.53\\
	&	{\tt RSA}	&	15	&	243	&	35	&	13	&	1,942	&  222,026 & 57,987 & 21,808 & 66,088 & 280,222 & 134,907.05 & 386.55\\
	&	{\tt SHA256}	&	3	&	159	&	36	&	2	&	4,992  & 5,574 & 1,040 & 243 & 1,024 & 7,532 & 3,201.07 & 440.67\\
\midrule													
\multirow{5}{*}{IWLS}	&	{\tt MEM\_CTRL}	&	27	&	492	&	442	&	160	&	2,096 &  4,007 & 1,051 & 120 & 1,136 & 5,183 & 2,373.35 & 260.72	\\
	&	{\tt SASC}	&	3	&	35	&	27	&	17	&	126	&  367 & 116 & 0 & 125 & 500 & 238.24 & 84.40 \\
	&	{\tt SIMPLE\_SPI	} &	3	&	55	&	34	&	15	&	288	&  476 & 130 & 2 & 145 & 623 & 282.57 & 119.42\\
	&	{\tt SS\_PCM}	&	1	&	5	&	10	&	3	&	24	&  231 & 87 & 1 & 94 & 338 & 168.29 & 90.51\\
	&	{\tt USB\_PHY}	&	3	&	67	&	70	&	34	&	223	&  287 & 98 & 0 & 85 & 401 & 194.15 & 71.91\\
\midrule													
\multirow{3}{*}{OpenCores}	&	{\tt ETHMAC}	&	66	&	487	&	1,217	&	218	&	3,849 &  34,783 & 10,545 & 2,195 & 12,021 & 45,441 &  22,453.76 & 190.44	\\
	&	{\tt I2C\_SLAVE} 	&	4	&	104	&	14	&	11	&	269	&  466 & 125 & 0 & 126 & 596 & 160.28 & 125.44\\
	&	{\tt VGA\_LCD}	&	16	&	123	&	310	&	56	&	885	&  54,614 & 17,052 & 4,921 & 19,228 & 71,766 & 36,095.90 & 224.67\\
\midrule													
\multirow{3}{*}{OpenROAD}	&	{\tt ARIANE\_ID} 	&	4	&	3,498	&	385	&	723	&	4,606 &  1,993 & 378 & 96 & 559 & 2,615 & 980.97 & 225.48\\
	&	{\tt GCD}	&	11	&	15	&	4	&	12	&	496	&  168 & 34 & 3 & 32 & 253 & 100.91 & 161.87\\
	&	{\tt IBEX}	&	15	&	14,740	&	5,815	&	6,330	&	26,885	&  12,161 & 1,864 & 978 & 2,965 & 14,379 & 5,758.84 & 538.10\\
\bottomrule
\end{tabularx}
\end{table*}

\begin{quote}
\leftskip =-20pt
\rightskip =-20pt
\textbf{Claim 1}: The unlocked CFGs $\mathcal{L}_{\mathcal{K}^*_r}(G(V,E), K_1)$ and $\mathcal{L}_{\mathcal{K}^*_r}(G(V,E), K_2$ using $r$-bit keys $K_1$ and $K_2$, respectively, on the XOR-based encrypted CFG $\mathcal{L}_{\mathcal{K}^*_r}(G(V,E),K)$ are unique.
\end{quote}

\begin{proof}
$\forall K_1 \neq K_2$, $K_1, K_2 \in \{0,1\}^r$ \\
$\implies K_1 \oplus V_{CN} \neq K_2 \oplus V_{CN}$ 
$\implies V^1_{CN} \neq V^2_{K_2}$. \\
$\implies G(V^1_{CN},E) \neq G(V^2_{CN},E)$.
\end{proof}

\begin{quote}
\leftskip =-20pt
\rightskip =-20pt
\textbf{Claim 2}: The obfuscated CFG $\mathcal{L}_{\mathcal{K}^*_r}(G(V,E),K)$ can be generated from $2^r$ possible combination of condition statuses with equal probability, i.e. the following condition holds true.
\begin{gather}     
P[\mathcal{L}_{\mathcal{K}'_r}(G(V,E),K) | G((V_{CN} \oplus K^*_r) \cup (V \setminus V_{CN}),E)] = \nonumber\\ 
P[\mathcal{L}_{\mathcal{K}'_r}(G(V,E),K) | G((V_{CN} \oplus K*_r) \cup (V \setminus V_{CN}),E)] \nonumber \\
    \forall K_r \neq K^*_r ; V^r_{CN} \neq V_{CN}
\end{gather}
\end{quote}
\begin{proof}
The probability of choosing $K_r$ is uniform. So, \\
$P[K = \mathcal{K}^*_r] = P[K = \mathcal{K}_r$], $\forall \mathcal{K}_r \neq \mathcal{K}^*_r, \mathcal{K}_r \in 2^r$ \\
$\implies P[(V_{CN}\oplus \mathcal{K}^*_r) \oplus \mathcal{K}^*_r] = P[(V_{CN}\oplus \mathcal{K}^*_r) \oplus \mathcal{K}_r]$ \\
$\implies P[V_{CN}] = P[V^r_{CN}],V_{CN} \neq V^r_{CN},$ \\
$V^r_{CN} = \{p_1,p_2,..,p_i,..,p_r\}$, where $p_i \in \{v_i,\overline{v_i}\}$.
\end{proof}
Combining claims 1 and 2 shows that the encrypted CFGs are indistinguishable for a family of $2^r$ possible designs.

In Fig.~\ref{fig:branchObfuscation_eg}, we report the area overhead of the two benchmark circuits \texttt{DES3} and \texttt{RSA} in case of branch obfuscation showing empirical evidence of our claim that obfuscated circuits are indistinguishable. We identify five conditions from each benchmark and generated five different variants, flipping each condition at a time. After synthesizing the circuits, we observed that area overhead is uniquely mapped to each variant of the design. The conditions in the CFG can be easily reverse engineered from the synthesized circuit and the flow of design can be unlocked. On the contrary, the area overhead of synthesized circuits remain the same after obfuscation, indicating the obfuscated circuits reveal no information about the control-flow of the circuit.

\section{Experimental Validation of ASSURE}
\subsection{Experimental Setup}

We implemented ASSURE as a Verilog$\rightarrow$Verilog tool that leverages Pyverilog~\cite{Takamaeda:2015:ARC:Pyverilog}, a Python-based 
hardware design processing toolkit to manipulate RTL Verilog. Pyverilog parses the input Verilog descriptions and creates the design AST. ASSURE then manipulates the AST. Pyverilog then reproduces the output Verilog description ready for logic synthesis.

\begin{figure*}
     \includegraphics[width=\textwidth]{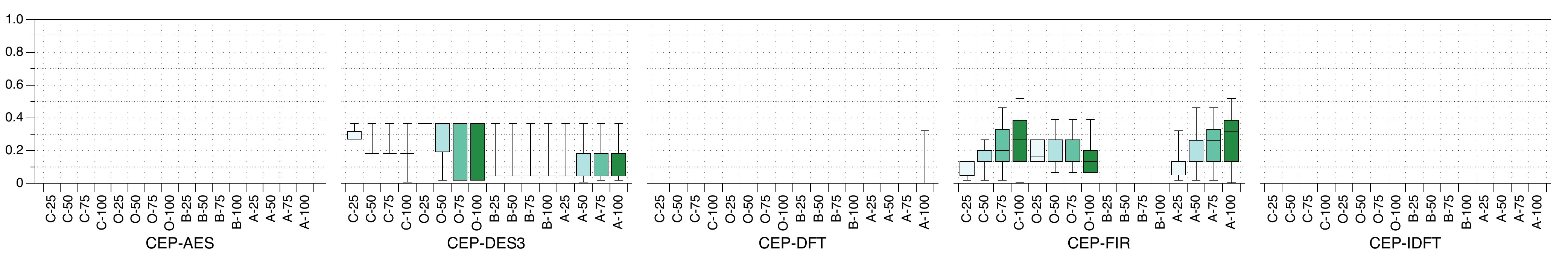}

     \includegraphics[width=\textwidth]{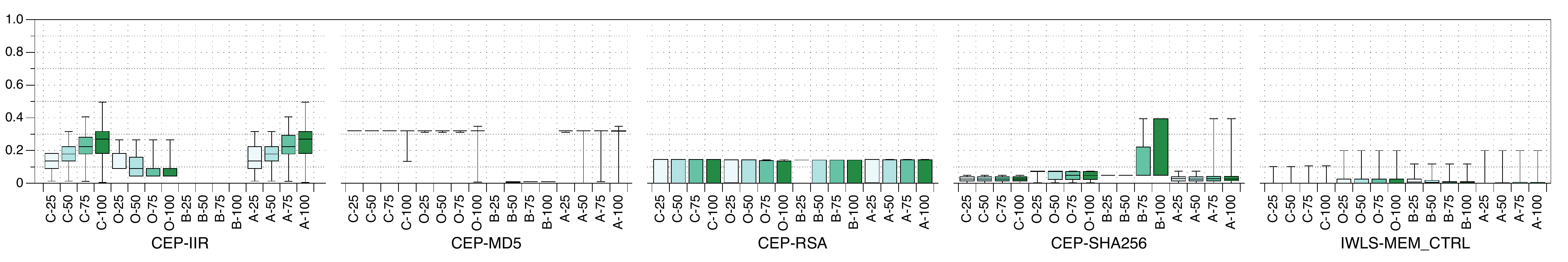}

     \includegraphics[width=\textwidth]{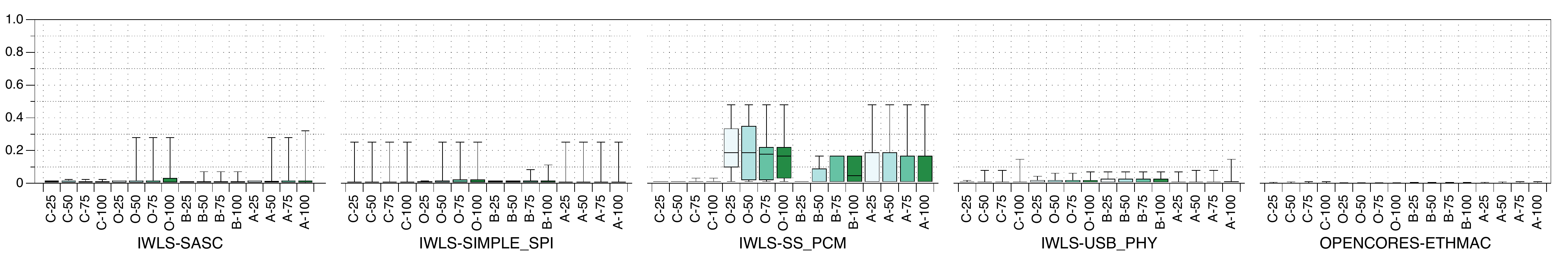}
     
     \includegraphics[width=\textwidth]{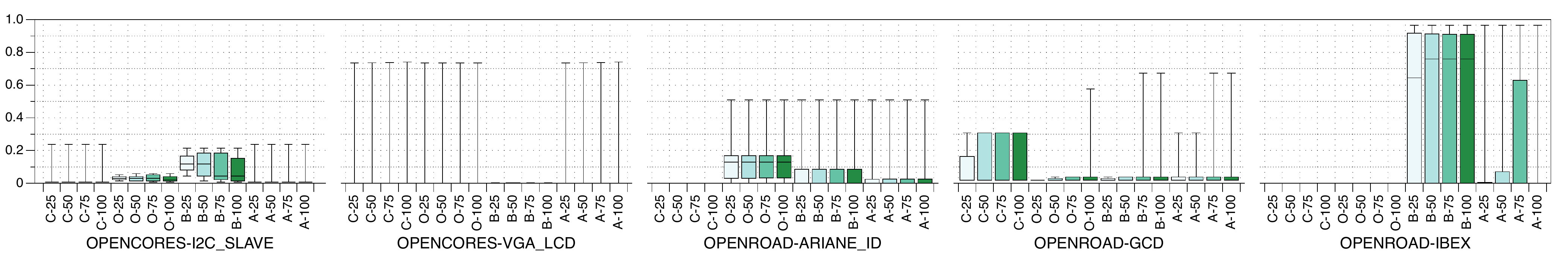}
     \caption{Verification failure metric in \textsc{KEY-EFFECT} experiments.}
     \label{fig:visibility_results}
\end{figure*}

We used ASSURE to protect several Verilog designs from different sources\footnote{Supporting VHDL and SystemVerilog only requires proper HDL parsers.}: the MIT-LL Common Evaluation Platform (CEP) platform~\cite{cep}, the OpenROAD project~\cite{openroad_dac2019}, and the OpenCores repository~\cite{opencores}. Four CEP benchmarks ({\tt DCT}, {\tt IDCT}, {\tt FIR}, {\tt IIR}) are created with Spiral, a hardware generator~\cite{Pueschel:11}). 
Table~\ref{tab:benchmarks} shows the characteristics of these benchmarks in terms of number of hardware modules, constants, operations, and branches. This data also characterizes the functionality that needs obfuscation. The benchmarks are much larger than those used by the gate-level logic locking experiments by the community~\cite{llcarxiv}. Differently from~\cite{dac18}, ASSURE does not require any modifications to synthesis tools and applies to pre-existing industrial designs, processing the Verilog RTL descriptions with no modifications. 

We analyzed the ASSURE in terms of security (Section~\ref{sec:security} and Section~\ref{sec:attacks}) and  overhead (Section~\ref{sec:cost}). For each benchmark, we created obfuscated variants using all techniques ({\tt ALL}) or one of constant ({\tt CONST}), operation {\tt (OP)}, and branch {\tt (BRANCH)} obfuscations. We repeat experiments by constraining the number of key bits available: 25\%, 50\%, 75\% or 100\% and reported in Table~\ref{tab:benchmarks}. The resulting design is then identified by a combination of its name, the configuration, and the number of key bits. For example, {\tt DFT-ALL-25} indicates obfuscation of the {\tt DFT} benchmark, where all three obfuscations are applied using 2,175 bits for obfuscation (25\% of 8,697) as follows: 38 for operations (25\% of 151), 33 for branches (25\% of 132) and the rest (2,104) for constants.

\subsection{Correctness and Key Effects}\label{sec:security}

We first apply formal verification on the locked design against the unprotected design with a twofold goal. First, we show that, when the correct key $\mathcal{K}^*_r$ is used, the unlocked circuit matches the original. We label this experiment as \textsc{correctness}. Second, we show that flipping each single key bit induces at least a failing point (i.e., no {\em collision}). This experiment demonstrates that each key bit has an effect on the functionality of the circuit. We label this experiment as \textsc{key effect}. We show no other key can activate the same IC, i.e., all other circuits ($K_r \neq K^*_r$) are not exact copies of the original designs. In this experiment, we also aim at quantifying how the obfuscation techniques affect the IC functionality when the attacker provides incorrect keys. With formal verification, we focus on IC functionality rather than IC results. We compute the {\em verification failure} metric as:
\begin{equation}
F = \frac{1}{K}\cdot\sum^{K}_{i=1}\frac{n(FailingPoints)_i}{n(TotalPoints)} 
\end{equation}
This metric is the average fraction of verification points that do not match when testing with different wrong keys. We experimented using Synopsys Formality N-2017.09-SP3.

\subsubsection{Correctness} We apply ASSURE several times, each time with a random key to obfuscate operations and branches (constants are always extracted in the same way. We formally verified these designs against the original ones. In all experiments, {\em\uline{ASSURE generates circuits that match the original design when using the correct key}}. 

\subsubsection{Key Effect} Given a design obfuscated with an $r$-bit key, we performed $r$ experiments where in each of them we flipped only one key bit with respect to the correct key. In all cases, {\em\uline{formal verification identifies at least one failing point, showing that an incorrect key always alters the circuit functionality}}. Varying the locking key has no effect since the failure is induced by the flipped bit (from correct to incorrect) and not its value.
Fig.~\ref{fig:visibility_results} shows the verification failure metrics for each experiment. Results are not reported for {\tt FIR-BRANCH-*} and {\tt IIR-BRANCH-*} because they have no branches. {\tt AES}, {\tt DFT}, {\tt IDFT}, and {\tt OPENCORES-ETHMAC} benchmarks have low values ($\sim$10$^{-5}$) since they have many verification points and only a small part is failing. Operations and constants vitally impact the design as obfuscating them induces more failing points. Increasing the number of obfuscated operations reduces the metric. Since obfuscation is performed using a depth-first analysis, the first bits correspond to operations closer to the inputs. As the analysis proceeds, obfuscation is closer to the output and more internal points match. 

\begin{figure}
    \centering
    \includegraphics[height=3.2cm]{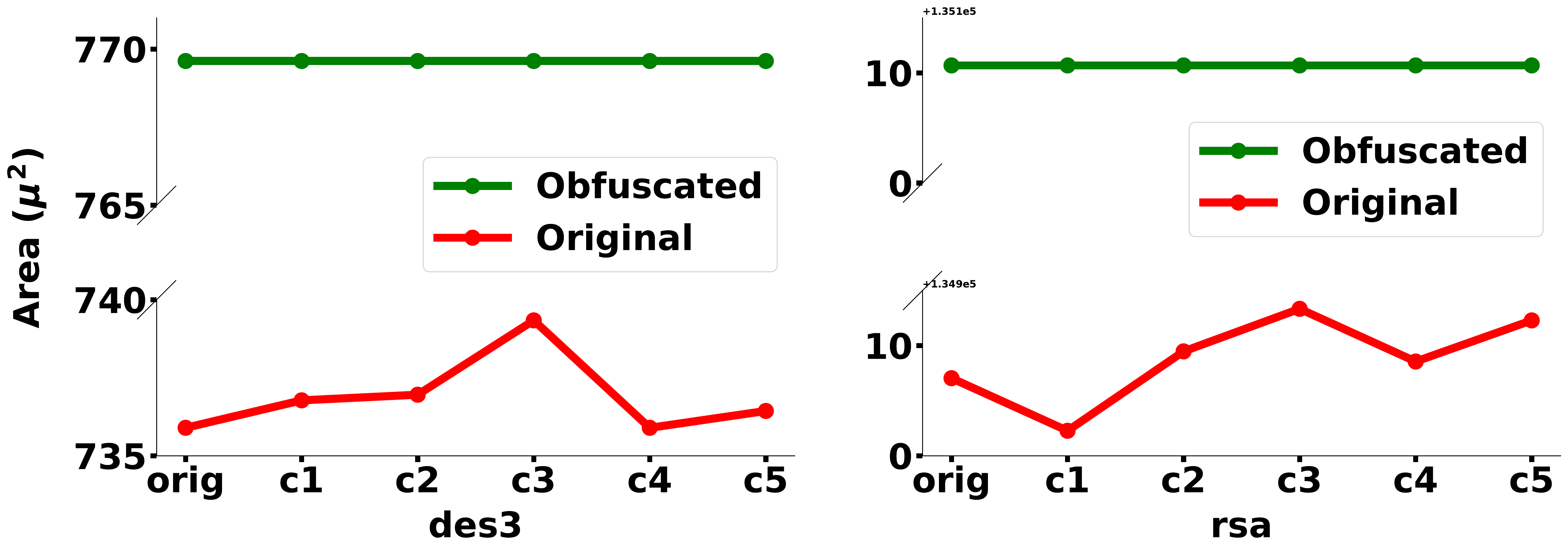}
    \caption{Area of original and obfuscated variants of  \texttt{DES3} and \texttt{RSA} when synthesized with different constants ($c1-c5$).}
    \label{fig:constantExtraction}
\end{figure}

This experiment allowed us to identify design practices that lead to inefficient obfuscations or even collisions. In {\tt DFT}, one-bit signals were initialized with integer values 0/1. Verilog allows this syntax and signals are trimmed by logic synthesis. A naive RTL constant analysis would pick 32 bits for obfuscating a single-bit. Since only the least significant bit impacts the circuit function, flipping the other 31 bits would lead to a collision. So, we extended ASSURE AST analysis to match the constant sizes with those of the target signals.

\begin{figure}[t]
    \centering
    \includegraphics[height=3.2cm]{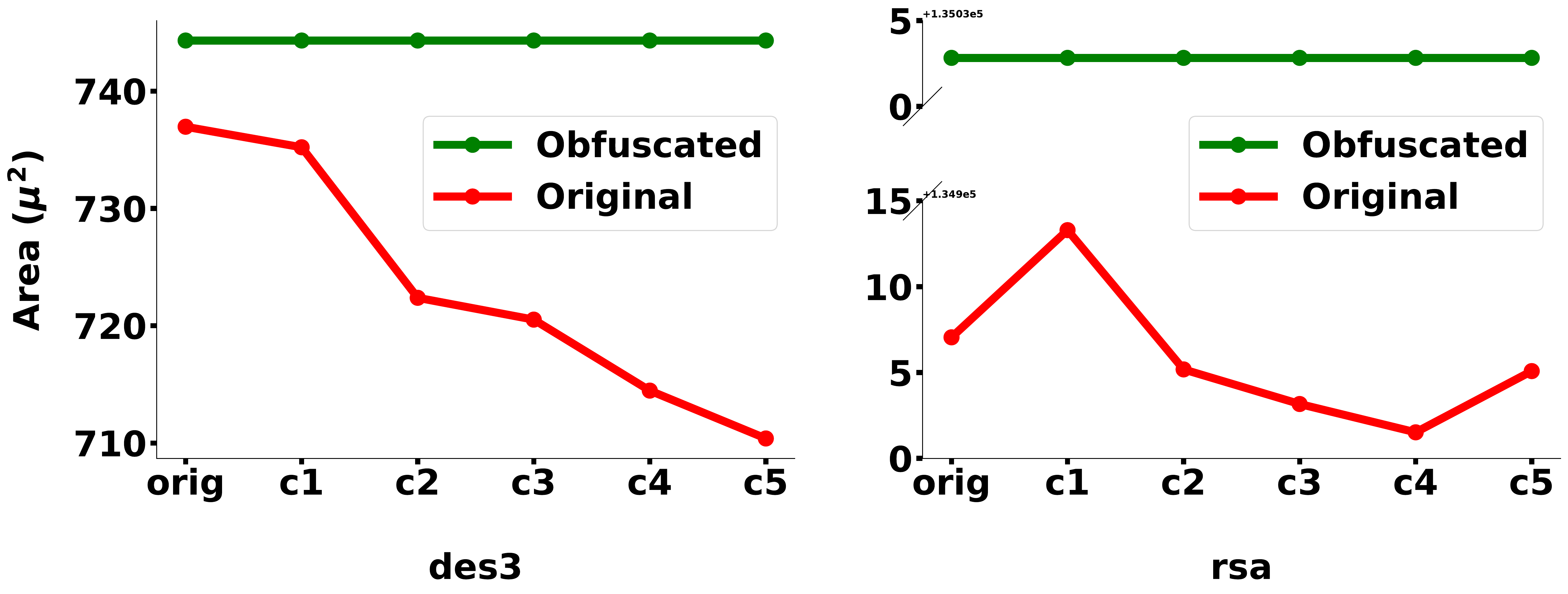}
    \caption{Area of original and obfuscated variants of benchmarks \texttt{DES3} and \texttt{RSA} using different operators in the statement.}
    \label{fig:opObfuscation_eg}
\end{figure}

\begin{table*}[!tb]
\caption{Security assessment of ASSURE obfuscation against redundancy attacks~\cite{8714955} (without oracle) and KC2 attacks~\cite{neos} (with oracle). {\em Failed} denotes attack failure due to non-existence of untestable faults for redundancy attacks and due to unsolvable constraints or incorrect key generation for KC2 attacks. {\em Timeout} denotes tool termination after \texttt{96} hours without returning the key. CFG1, CFG2, CFG3, and CFG4 corresponds to 25\%, 50\%, 75\%, and 100\% of maximum possible key-bit obfuscation, respectively. (X/Y) in redundancy attack indicate X keybits are correct out of Y keybits recovered by the redundancy attack.}
\resizebox{2\columnwidth}{!}{
\begin{tabular}{ccccccccccccccc}
\hline
\multirow{3}{*}{\textbf{\begin{tabular}[c]{@{}c@{}}Bench\\ mark\end{tabular}}} & \multirow{3}{*}{\textbf{\begin{tabular}[c]{@{}c@{}}Obf.\\ Type\end{tabular}}} & \multirow{3}{*}{\textbf{\begin{tabular}[c]{@{}c@{}}Attack \\ with \\oracle\\ access?\end{tabular}}} & \multicolumn{12}{c}{\textbf{Obfuscation configuration}} \\
    &     &   & \multicolumn{3}{c}{\textbf{CFG1}}  & \multicolumn{3}{c}{\textbf{CFG2}} & \multicolumn{3}{c}{\textbf{CFG3}}  & \multicolumn{3}{c}{\textbf{CFG4}}  \\ \cline{4-15} 
    &     &    & \textbf{\begin{tabular}[c]{@{}c@{}}Key\\ (bits)\end{tabular}} & \textbf{\begin{tabular}[c]{@{}c@{}}Recovered\\ (bits)\end{tabular}} & \textbf{\begin{tabular}[c]{@{}c@{}}Time \\ (s)\end{tabular}} & \textbf{\begin{tabular}[c]{@{}c@{}}Key\\ (bits)\end{tabular}} & \textbf{\begin{tabular}[c]{@{}c@{}}Recovered\\ (bits)\end{tabular}} & \textbf{\begin{tabular}[c]{@{}c@{}}Time \\ (s)\end{tabular}} & \textbf{\begin{tabular}[c]{@{}c@{}}Key\\ (bits)\end{tabular}} & \textbf{\begin{tabular}[c]{@{}c@{}}Recovered\\ (bits)\end{tabular}} & \textbf{\begin{tabular}[c]{@{}c@{}}Time \\ (s)\end{tabular}} & \textbf{\begin{tabular}[c]{@{}c@{}}Key\\ (bits)\end{tabular}} & \textbf{\begin{tabular}[c]{@{}c@{}}Recovered\\ (bits)\end{tabular}} & \textbf{\begin{tabular}[c]{@{}c@{}}Time\\ (s)\end{tabular}} \\ \hline
\multirow{4}{*}{{\tt DES3}} 
& \multirow{2}{*}{All} & no & \multirow{2}{*}{225}  & 20/34 & 5,655 & \multirow{2}{*}{450} & 31/54  & 20,860  
& \multirow{2}{*}{675} & 0  & timeout & \multirow{2}{*}{900} & 0 & timeout \\
&                      & yes    &   & 225 & 13,447 &   & 450  & 16,216 &     & 0  & failed &  & 0  & timeout  \\
& \multirow{2}{*}{Constant} & no & \multirow{2}{*}{30} & 0/8 & 264                                            & \multirow{2}{*}{60} & 0/8 & 968 & \multirow{2}{*}{90} & 0/10 & 1,456 & \multirow{2}{*}{120} & 0/10 & 2,575 \\
&                      & yes    &   & 30  & 2,324 &    & 60 & 5,398 &   & 0  & failed &      & 120  & 8,476 \\ \hline
\multirow{4}{*}{{\tt FIR}}                                                  & \multirow{2}{*}{All}                                                                     & no                                                                                 & \multirow{2}{*}{86}                                           & 4/32                                                                & 3,269                                                         & \multirow{2}{*}{164}                                          & 7/45                                                                & 26,045                                                        & \multirow{2}{*}{250}                                          & 12/67                                                               & 39,025                                                        & \multirow{2}{*}{336}                                          & 0                                                                   & timeout                                                     \\
                                                                               &                                                                                          & yes                                                                                &                                                               & 0                                                                   & 1,372                                                         &                                                               & 0                                                                   & failed                                                       &                                                               & 0                                                                   & 5,665                                                         &                                                               & 0                                                                   & timeout                                                     \\
                                                                               & \multirow{2}{*}{Constant}                                                                & no                                                                                 & \multirow{2}{*}{80}                                           & 0/25                                                                & 2,989                                                         & \multirow{2}{*}{152}                                          & 0/26                                                                & 22,697                                                        & \multirow{2}{*}{232}                                          & 0/52                                                                & 33,156                                                        & \multirow{2}{*}{312}                                          & 0                                                                   & timeout                                                     \\
                                                                               &                                                                                          & yes                                                                                &                                                               & 0                                                                   & 1,189                                                         &                                                               & 0                                                                   & failed                                                       &                                                               & 0                                                                   & 5,145                                                         &                                                               & 0                                                                   & timeout                                                     \\ \hline
\multirow{4}{*}{{\tt MD5}}                                                  & \multirow{2}{*}{All}                                                                     & no                                                                                 & \multirow{2}{*}{1,135}                                         & 0                                                                   & timeout                                                      & \multirow{2}{*}{2,267}                                         & 0                                                                   & timeout                                                      & \multirow{2}{*}{3,401}                                         & 0                                                                   & timeout                                                      & \multirow{2}{*}{4,533}                                         & 0                                                                   & timeout                                                     \\
                                                                               &                                                                                          & yes                                                                                &                                                               & 0                                                                   & failed                                                       &                                                               & 0                                                                   & timeout                                                      &                                                               & 0                                                                   & timeout                                                      &                                                               & 0                                                                   & timeout                                                     \\
                                                                               & \multirow{2}{*}{Constant}                                                                & no                                                                                 & \multirow{2}{*}{1,121}                                         & 0                                                                   & timeout                                                      & \multirow{2}{*}{2,241}                                         & 0                                                                   & timeout                                                      & \multirow{2}{*}{3,362}                                         & 0                                                                   & timeout                                                      & \multirow{2}{*}{4,482}                                         & 0                                                                   & timeout                                                     \\
                                                                               &                                                                                          & yes                                                                                &                                                               & 0                                                                   & failed                                                       &                                                               & 0                                                                   & timeout                                                      &                                                               & 0                                                                   & timeout                                                      &                                                               & 0                                                                   & timeout                                                     \\ \hline
\multirow{4}{*}{{\tt SHA256}}                                               & \multirow{2}{*}{All}                                                                     & no                                                                                 & \multirow{2}{*}{1,250}                                         & 0                                                                   & timeout                                                      & \multirow{2}{*}{2,496}                                         & 0                                                                   & timeout                                                      & \multirow{2}{*}{3,745}                                         & 0                                                                   & timeout                                                      & \multirow{2}{*}{4,992}                                         & 0                                                                   & timeout                                                     \\
                                                                               &                                                                                          & yes                                                                                &                                                               & 0                                                                   & failed                                                       &                                                               & 0                                                                   & failed                                                       &                                                               & 0                                                                   & timeout                                                      &                                                               & 0                                                                   & timeout                                                     \\
                                                                               & \multirow{2}{*}{Constant}                                                                & no                                                                                 & \multirow{2}{*}{1,239}                                         & 0                                                                   & timeout                                                      & \multirow{2}{*}{2,477}                                         & 0                                                                   & timeout                                                      & \multirow{2}{*}{3,716}                                         & 0                                                                   & timeout                                                      & \multirow{2}{*}{4,954}                                         & 0                                                                   & timeout                                                     \\
                                                                               &                                                                                          & yes                                                                                &                                                               & 0                                                                   & failed                                                       &                                                               & 0                                                                   & failed                                                       &                                                               & 0                                                                   & timeout                                                      &                                                               & 0                                                                   & timeout                                                     \\ \hline
\multirow{4}{*}{{\tt SS\_PCM}}                                              & \multirow{2}{*}{All}                                                                     & no                                                                                 & \multirow{2}{*}{7}                                            & 0/4                                                                 & 2                                                            & \multirow{2}{*}{13}                                           & 0/4                                                                 & 3                                                            & \multirow{2}{*}{18}                                           & 1/5                                                                 & 5                                                            & \multirow{2}{*}{24}                                           & 1/5                                                                 & 7                                                           \\
                                                                               &                                                                                          & yes                                                                                &                                                               & 7                                                                   & 843                                                          &                                                               & 13                                                                  & 170                                                          &                                                               & 18                                                                  & 1,308                                                         &                                                               & 0                                                                   & 6,052                                                        \\
                                                                               & \multirow{2}{*}{Constant}                                                                & no                                                                                 & \multirow{2}{*}{3}                                            & 0/0                                                                 & 2                                                            & \multirow{2}{*}{6}                                            & 0/0                                                                 & 2                                                            & \multirow{2}{*}{8}                                            & 0/0                                                                 & 3                                                            & \multirow{2}{*}{11}                                           & 0/0                                                                 & 5                                                           \\
                                                                               &                                                                                          & yes                                                                                &                                                               & 3                                                                   & 289                                                          &                                                               & 6                                                                   & 310                                                          &                                                               & 8                                                                   & 784                                                          &                                                               & 0                                                                   & 1897                                                        \\ \hline
\multirow{4}{*}{{\tt GCD}}                                                  & \multirow{2}{*}{All}                                                                     & no                                                                                 & \multirow{2}{*}{11}                                           & 3/11                                                                & 8                                                            & \multirow{2}{*}{23}                                           & 5/15                                                                & 8                                                            & \multirow{2}{*}{34}                                           & 7/17                                                                & 12                                                           & \multirow{2}{*}{47}                                           & 9/16                                                                & 14                                                          \\
                                                                               &                                                                                          & yes   &     & 0  & 8                                                            &                                                               &                                                                  0   & 15                                                           &                                                               &                                                                   0  & 15                                                           &                                                               &                                                                 0    & 21                                                          \\
                                                                               & \multirow{2}{*}{Constant}                                                                & no                                                                                 & \multirow{2}{*}{7}                                            & 0/0                                                                 & 6                                                            & \multirow{2}{*}{15}                                           & 0/4                                                                 & 7                                                            & \multirow{2}{*}{22}                                           & 0/8                                                                 & 11                                                           & \multirow{2}{*}{31}                                           & 0/8                                                                 & 14                                                          \\
                                                                               &                                                                                          & yes                                                                                &                                                               &                                                                 0    & 7                                                            &                                                               &                                                               0      & 7                                                            &                                                               &                                                                 0    & 14                                                           &                                                               &                                                                  0   & 19                                                          \\ \hline
\multirow{4}{*}{{\begin{tabular}[c]{@{}c@{}}\tt USB\_\\ \tt PHY\end{tabular}}}  & \multirow{2}{*}{All}                                                                     & no                                                                                 & \multirow{2}{*}{57}                                           & 15/21                                                               & 17                                                           & \multirow{2}{*}{112}                                          & 0                                                                   & failed                                                       & \multirow{2}{*}{163}                                          & 34/75                                                               & 105                                                          & \multirow{2}{*}{223}                                          & 47/86                                                               & 184                                                         \\
                                                                               &                                                                                          & yes                                                                                &                                                               & 0                                                                   & 521                                                          &                                                               & 0                                                                   & 548                                                          &                                                               & 0                                                                   & 898                                                          &                                                               & 0                                                                   & 360                                                         \\
                                                                               & \multirow{2}{*}{Constant}                                                                & no                                                                                 & \multirow{2}{*}{30}                                           & 0/0                                                                 & 14                                                           & \multirow{2}{*}{60}                                           & 0                                                                   & failed                                                       & \multirow{2}{*}{89}                                           & 0/5                                                                 & 97                                                           & \multirow{2}{*}{119}                                          & 0/10                                                                & 152                                                         \\
                                                                               &                                                                                          & yes                                                                                &                                                               & 0                                                                   & 510                                                          &                                                               & 0                                                                   & 522                                                          &                                                               & 0                                                                   & 524                                                          &                                                               & 0                                                                   & 347                                                         \\ \hline
\end{tabular}
}
\label{tab2:attacksOnAssure}
\end{table*}

\subsection{Resilience Against Locking Attacks}\label{sec:attacks}

We outlined provable security guarantees that $n$ obfuscation bits induce $2^n$  RTL designs with uniform probability. We now discuss resilience to known locking attacks.

\subsubsection{Resynthesis Attacks}\label{sec:resynthesis}
Massad \textit{et al.}~\cite{DBLP:journals/corr/MassadZGT17} showed that greedy heuristics can recover the key of an obfuscated gate-level netlist. Performing re-synthesis with an incorrect key may trigger additional optimizations that produce large redundancy in the circuit. Similarly, Li \textit{et al.}~\cite{8714955} propose an attack using concepts from VLSI testing. Incorrect key results in large logic redundancy and most of stuck-at faults become untestable. A correctly unlocked circuit, however, has high testability. ASSURE obfuscates RTL design before synthesis. Since the obfuscated RTL is equally likely to be generated from $2^n$ designs, logic synthesis using different keys on a reverse-engineered obfuscated netlist reveals no information about the original netlist. Hence, the area overhead for the correct and incorrect keys are in same range (see Figs.~\ref{fig:constantExtraction}, \ref{fig:opObfuscation_eg} and \ref{fig:branchObfuscation_eg}).

\begin{figure}
    \centering
    \includegraphics[height=3.2cm]{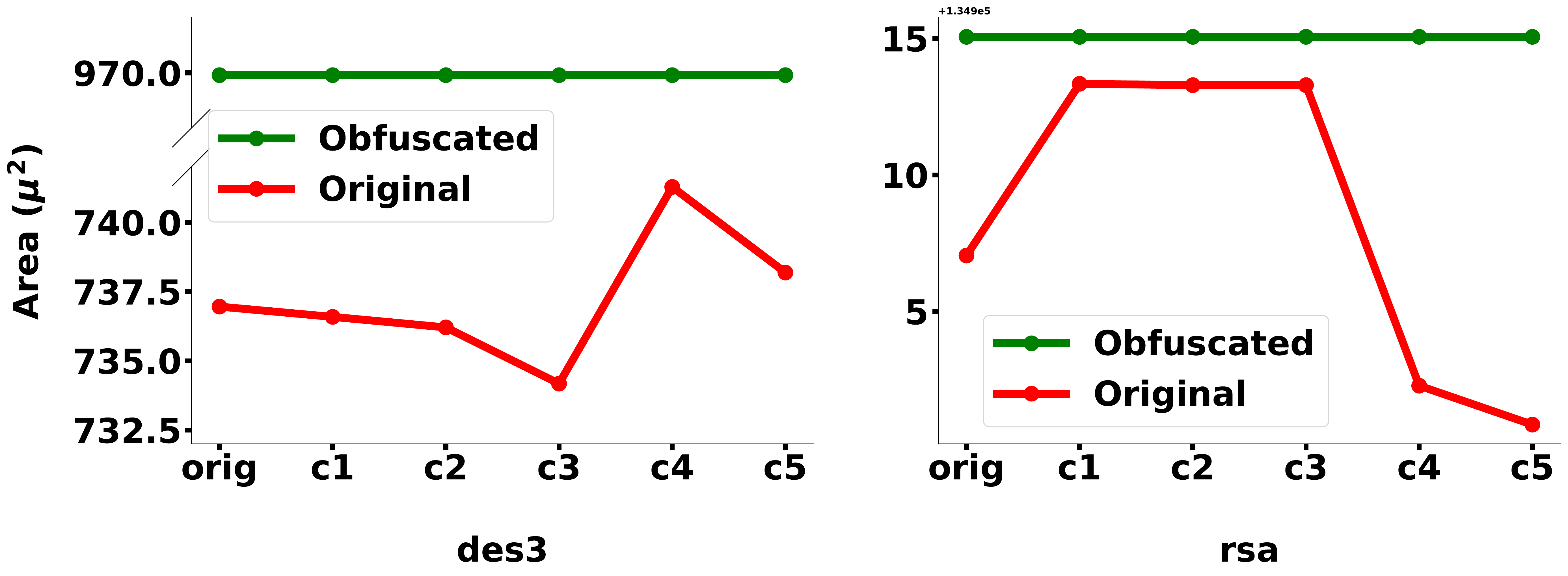}
    \caption{Area of original and obfuscated variants of benchmarks \texttt{DES3} and \texttt{RSA} in case of different CFG flows.}
    \label{fig:branchObfuscation_eg}
\end{figure}

\subsubsection{ML-guided Attacks}\label{sec:mlguided_attacks} 
Chakraborty \textit{et al.}~\cite{sail,surf} proposed oracle-less attacks on logic obfuscation based on the idea that obfuscation techniques insert XOR/XNOR gates that leave structural traces. The key gates are inserted before synthesis with known technology library and synthesis process (algorithms and tools). Since the effect of logic optimizations are local and the optimization rules are deterministic, one can recover the original function by launching an ML-guided removal attack on the obfuscated RTL design. In ASSURE, the obfuscation logic does not depend solely on the insertion of XOR/XNOR gates. For example, in branch obfuscation, we perform also logic inversion of the condition instead of simply adding a XOR gate followed by a NOT when the corresponding key bit is $1$. Recovering the original RTL from obfuscated RTL is hard (see claim 2 of ASSURE branch obfuscation proof in Section~\ref{sec:branch_obf}). Also, recovering extracted constants from an obfuscated design is impossible since the obfuscated circuit does not contain any information on the constant value.

\subsubsection{Redundancy and KC2 Attacks}\label{sec:other_attacks}
We analyze the strength of ASSURE's obfuscation by running oracle-less redundancy attacks~\cite{8714955}. Redundancy attacks decipher the key bits by identifying redundant lines in the synthesized netlist with incorrect key bits. KC2 is an improved version of SAT-based attacks incrementally unrolling a sequential circuit to recover the key. 
Even if ASSURE is not designed to protect against oracle-guided attacks, we evaluated its performance on KC2~\cite{neos}, a popular oracle-guided attack. 
We have run both attacks with a timeout of 96 hours and 50~GB of memory for each attack run. Table~\ref{tab2:attacksOnAssure} summarizes the results of both attacks on selected ASSURE obfuscated designs. In particular, we apply the attacks to benchmarks that we can safely convert into the format required by the attack tools.
We perform the attacks after applying all obfuscations ({\tt ALL}) or after applying only constant obfuscation ({\tt Constant}).
Constant obfuscation successfully thwarts all redundancy attacks showing this is the most powerful obfuscation. Indeed, on benchmarks like {\tt DES3}, {\tt FIR}, {\tt SS\_PCM}, and {\tt USB\_PHY}, redundancy attacks recovered some key bits. These results indicate that combining different obfuscation techniques is not 100\% secure compared to stand-alone obfuscation.
Even if we focused on the netlist-only threat model, it is interesting to evaluate the effects of oracle-guided attacks. KC2 attacks were able to recover the correct keys only for {\tt DES3} and {\tt SS\_PCM} benchmarks. 
In all other cases, KC2 claimed to recover certain key bits. However, the equivalence checking performed by ABC within the tool showed the functionality unlocked with those bits is not equivalent with the original one, i.e. the key is incorrect.

\begin{figure*}
     \centering
     \includegraphics[width=\textwidth]{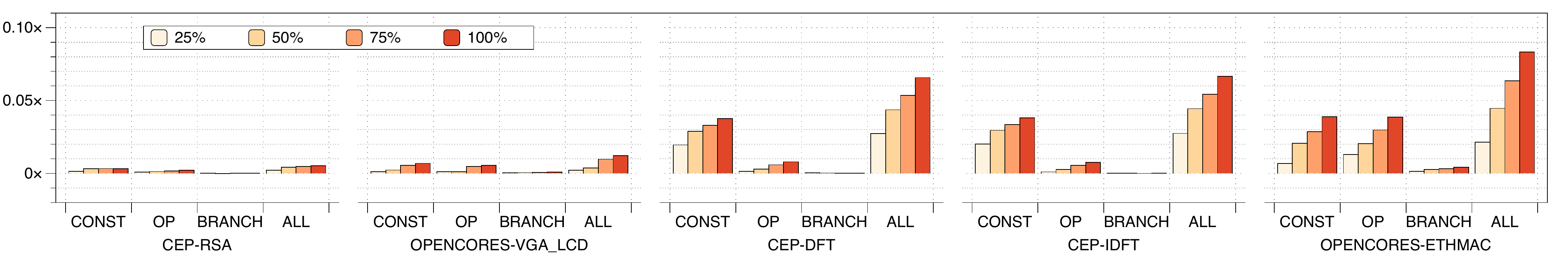}

     \includegraphics[width=\textwidth]{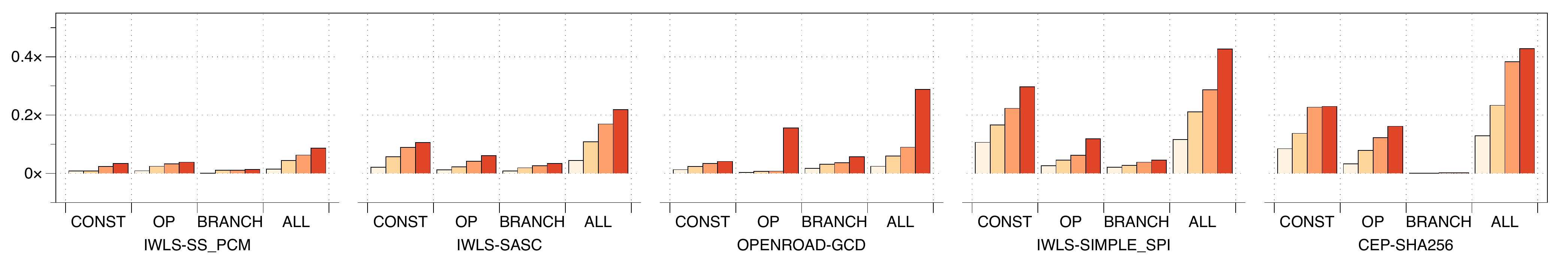}

     \includegraphics[width=\textwidth]{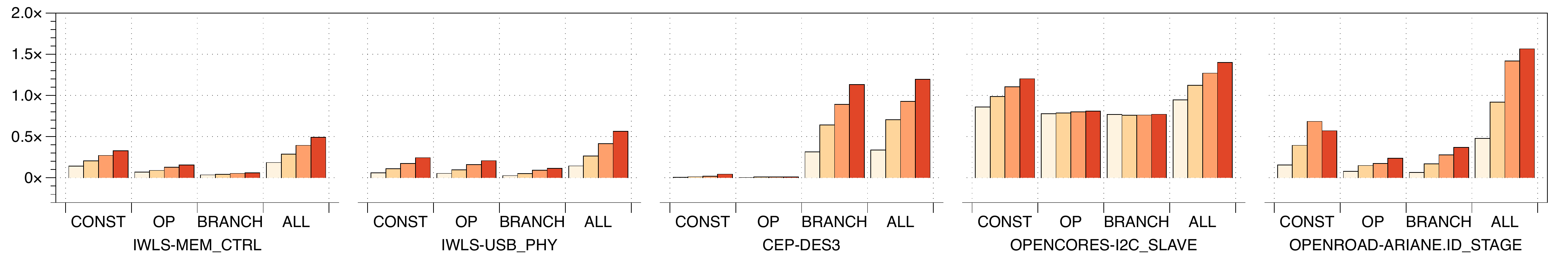}
     
     \includegraphics[width=\textwidth]{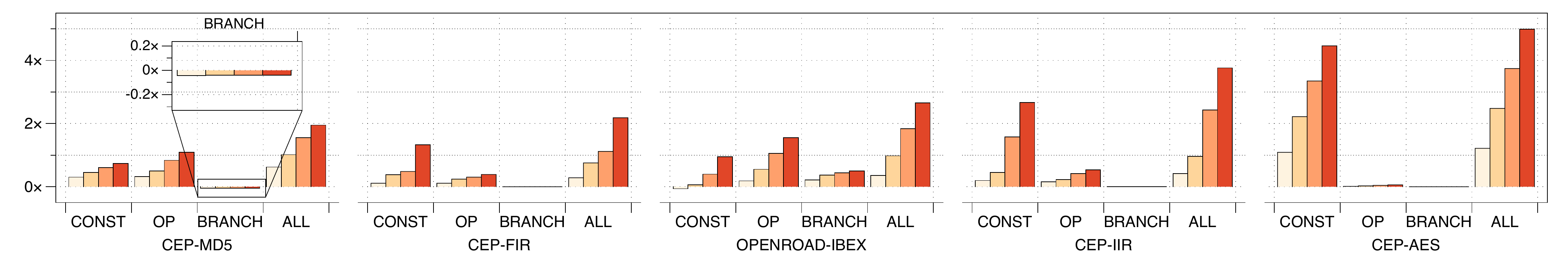}
\caption{\vspace{-4pt}Area overhead for ASSURE obfuscation. Benchmarks are presented in increasing order of total overhead.}
     \label{fig:area_overhead_key_bit}
\end{figure*}

\begin{figure*}
     \centering
     \includegraphics[width=0.98\columnwidth]{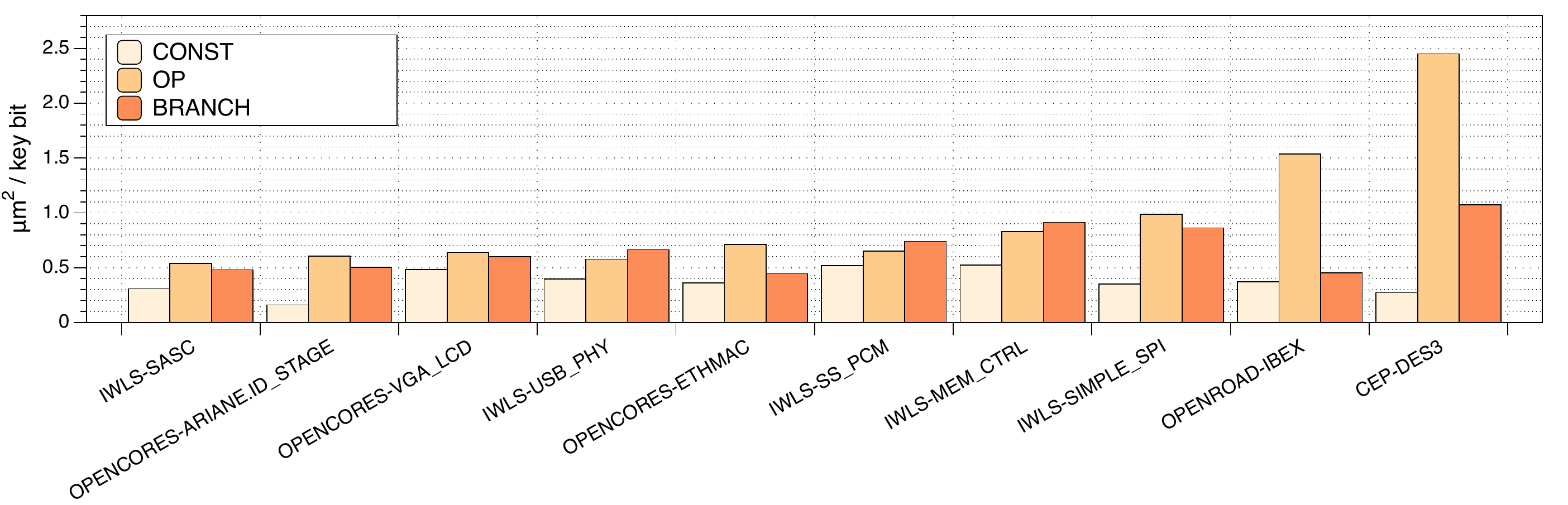}
\hfill     
     \includegraphics[width=0.98\columnwidth]{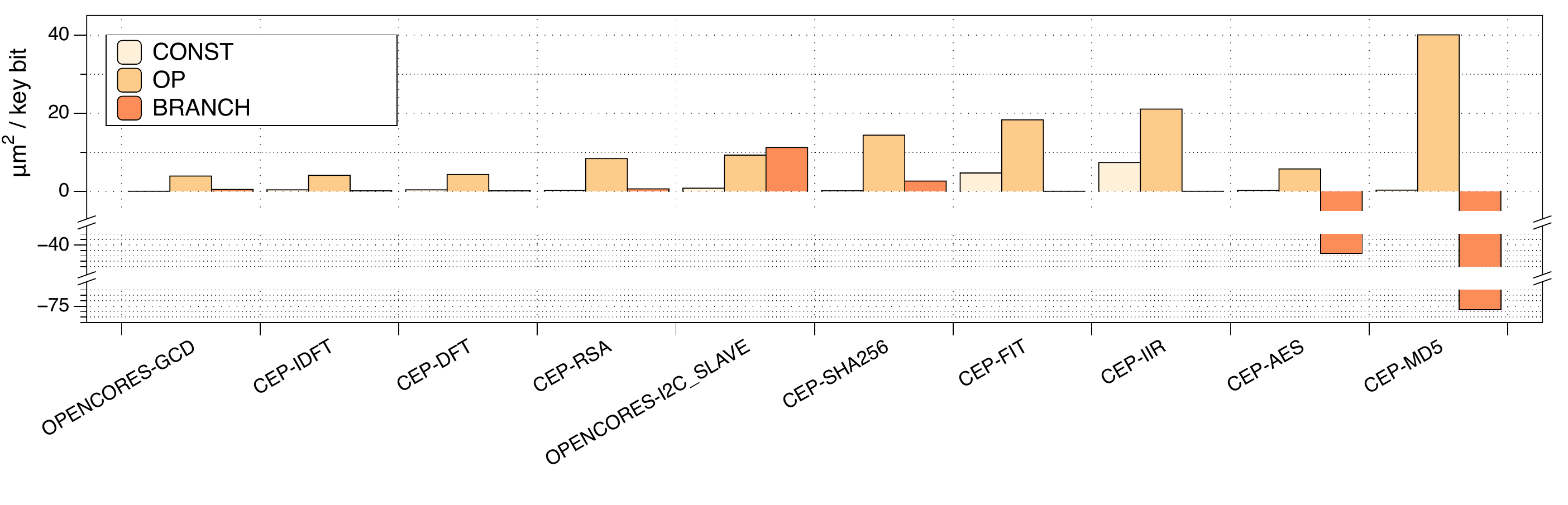}
\vspace{-4pt}\caption{Area overhead per key bit for ASSURE obfuscation. Benchmarks are presented in increasing order of total overhead.}
     \label{fig:area_overhead}
\end{figure*}

\begin{figure*}
     \centering
     \includegraphics[width=\textwidth]{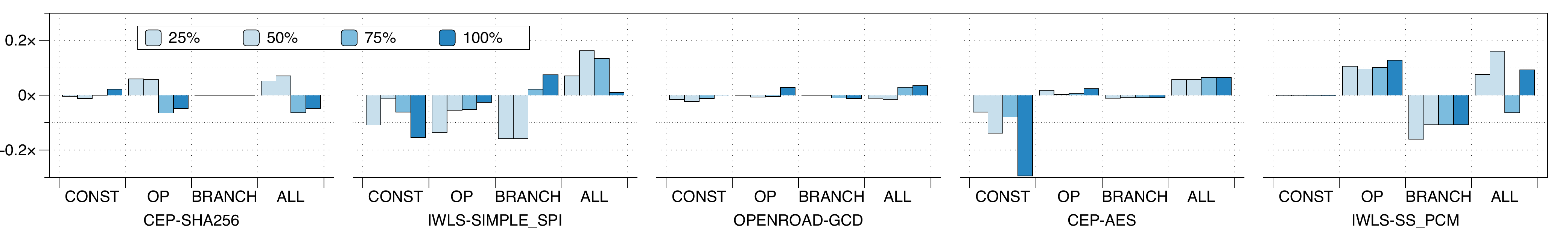}
     
     \includegraphics[width=\textwidth]{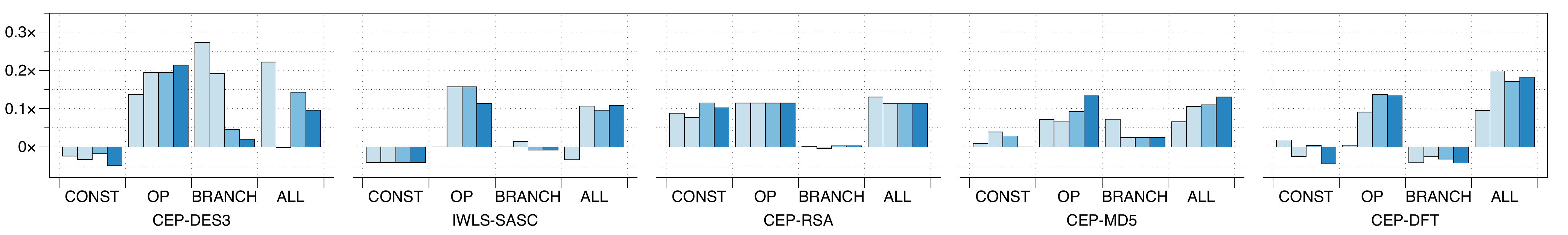}

     \includegraphics[width=\textwidth]{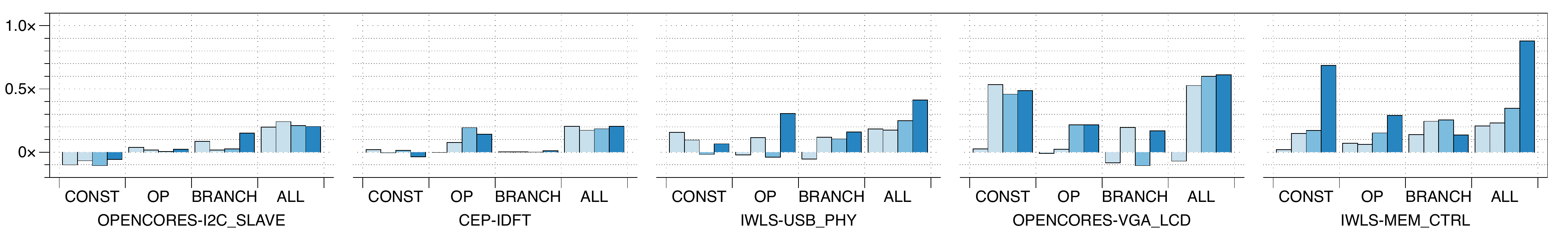}     

     \includegraphics[width=\textwidth]{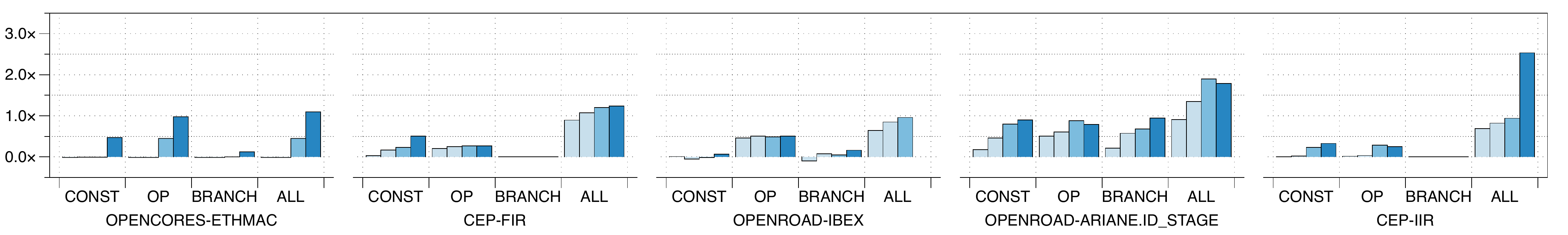}
\caption{     \vspace{-4pt}Timing overhead for ASSURE obfuscation. Benchmarks are presented in increasing  order of total overhead.}
     \label{fig:timing_overhead}
\end{figure*}

\subsection{Synthesis Overhead}\label{sec:cost}
We did logic synthesis using the Synopsys Design Compiler J-2018.04-SP5 targeting the Nangate 15nm ASIC technology at standard operating conditions ($25^{\degree}$C). We evaluated the area overhead and critical-path delay degradation relative  to the original design. While our goal is to protect the IP functionality and not to optimize the  resources, designs with lower cost are preferred.  {\em\uline{ASSURE generates correct designs with no combinational loops}}. Constant obfuscation extracts the values that are used as the key and no extra logic. Operation obfuscation multiplexes results of original and dummy operations. Branch obfuscation adds XOR to the conditions. 

\subsubsection{Area overhead}\label{sec:area_cost}
Table~\ref{tab:benchmarks} reports the results of the original design -- the number of cells in the netlists, the area (in $\mu m^2$) and the critical-path delay (in $ns$). Fig.~\ref{fig:area_overhead} reports the area overhead of all obfuscations with respect to the original designs. The three techniques are independent and so, {\tt ALL} results are the aggregate of the three techniques. 
Constant obfuscation produces an average overhead in the range 18\% ({\tt *-CONST-25}) to 80\% ({\tt *-CONST-100}). The maximum overhead is about 450\% for {\tt AES-CONST-100}, which has the most number obfuscated constants. {\em\uline{ASSURE removes hard-coded constants from the circuit, preventing logic optimizations like constant propagation}}. In {\tt AES}, all S-Box modules are optimized as logic in the original circuit. This optimization is not possible anymore when the constants are provided as inputs. However, we showed this obfuscation provides maximum protection since the constants are semantically removed from the circuit. The average operation obfuscation overhead is in the range 9\% ({\tt *-OP-25}) to 25\% ({\tt *-OP-100}). {\tt IBEX-OP-100} has the maximum overhead of 155\% since it has the most operations. Branch obfuscation produces a smaller average overhead, in the range 6\% ({\tt *-BRANCH-25}) to 14\% ({\tt *-BRANCH-100}) with a maximum overhead of 113\% for {\tt DES-BRANCH-100}. This benchmark has the largest proportion of branches relative to other elements. {\tt MD5} results in savings ($\sim$4\%) when we apply branch obfuscation ({\tt MD5-BRANCH-*}). The branch conditions help  pick elements from the library that lower area overhead.

The real impact of ASSURE depends on how many elements are obfuscated in each configuration. So, we computed the {\em area overhead per key bit} as the area overhead of a configuration divided by the number of key bits used for its obfuscation and report it in Fig.~\ref{fig:area_overhead_key_bit}.  
In most cases, {\em\uline{operation obfuscation has the largest impact, followed by branches and then constants}}. This impact is larger for data-intensive benchmarks, like CEP filters ({\tt DFT}, {\tt IDFT}, {\tt FIR}, and {\tt IIR}). Constants usually require more obfuscation bits, so the impact per bit is smaller. Each obfuscated operation introduces a new functional unit and multiplexer per key bit. {\tt MD5} has a large negative impact when obfuscating the branches justifying the area reduction when we apply only branch obfuscation ({\tt MD5-BRANCH-*}). On the contrary, even if {\tt AES} was the benchmark with the largest overhead (and many more bits), its overhead per key bit is comparable with the others. We repeated the experiments several times and we observed minimal variants with different locking keys.

To conclude {\em\uline{the area overhead is related to the design characteristics and to the number of key bits}}. The former determine the impact of  ASSURE, while the latter determine the total amount of overhead. {\em\uline{The overhead depends on the design, the techniques, and the number of key bits and not on the values of the locking key}}.

\subsubsection{Timing overhead}\label{sec:timing_cost}
Fig.~\ref{fig:timing_overhead} shows the overhead introduced by the ASSURE obfuscation logic on the critical path when targeting area optimization. {\em\uline{Timing overhead is application dependent with similar results across the different techniques. The overhead is larger when the obfuscated elements are on the critical path}}. This is relevant in data-intensive (with many operations) and control-intensive (with control branches on critical path) designs. In most benchmarks, the timing overhead is $<$20\%. Constants have a positive impact on the overhead (see {\tt AES} and {\tt DES3}). The obfuscated designs can generally achieve the same performance as the original ones, limiting the impact on the IC design flow.

\section{Discussion and Concluding Remarks}\label{sec:discussion}

We presented ASSURE, an RTL locking framework against an {\em\uline{untrusted foundry that has no access to an unlocked functional chip}}. ASSURE operates on the Verilog RTL description and is fully compatible with industrial EDA flows. We discuss the major contributions in the form of Q\&A.

\noindent $\blacktriangleright$ {\bf Which threat model are you considering? How is it relevant for my design?} We consider the netlist-only threat model where the attacker has no access to an activated chip. This model is relevant especially for an untrusted foundry with low-volume IC production. 

\noindent $\blacktriangleright$ {\bf Why should I use an RTL approach instead of existing gate-level techniques?}
ASSURE hides the essential semantics (constants, operations, and control-flow branches) in a way that is {\em\uline{indistinguishible and provably secure against attackers with no prior knowledge of the IP function}}. Most of the semantic information (e.g., constants) cannot be protected at the gate level because synthesis tools embed it into the netlist.

\noindent $\blacktriangleright$ {\bf Is ASSURE secure?} In our experimental analysis with formal verification and logic synthesis EDA tools, we show the {\em\uline{circuits can be unlocked only with the correct key}} and {\em\uline{obfuscating the design closer to the inputs induces more verification failures}}. Also, {\em\uline{ASSURE can thwart oracle-less attacks that can recover key bits even in case of SAT-resilient protections}}, showing the two approaches must be combined.

\noindent $\blacktriangleright$ {\bf What is the overhead?} ASSURE obfuscations introduce {\em\uline{area overhead that depends on the obfuscation techniques and is proportional to the number of key bits}}. In case of constants, obfuscation prevent logic optimizations, like constant propagation, while {\em\uline{operation obfuscation has the largest overhead per key bit}}. The {\em\uline{key values have no impact on the obfuscation results}}. ASSURE has no impact on the clock cycles but only on the critical path delay in a way that depends on where the obfuscation is applied. The designers can use these guidelines to apply obfuscation on their design.

\section*{Acknowledments}

The authors would like to thank Benjamin Tan (NYU) and Jitendra Bhandari (NYU) for their support in implementing locking attacks. The research is supported in part by NSF Award (\# 1526405), ONR Award (\# N00014-18-1-2058), NSF CAREER Award (\# 1553419), the NYU Center for Cybersecurity (\url{cyber.nyu.edu}), and the NYUAD Center for Cybersecurity (\url{sites.nyuad.nyu.edu/ccs-ad}).

\bibliographystyle{IEEEtran}
\bibliography{assure}

\begin{IEEEbiography}[{\includegraphics[width=1in,height=1.25in,clip,keepaspectratio]{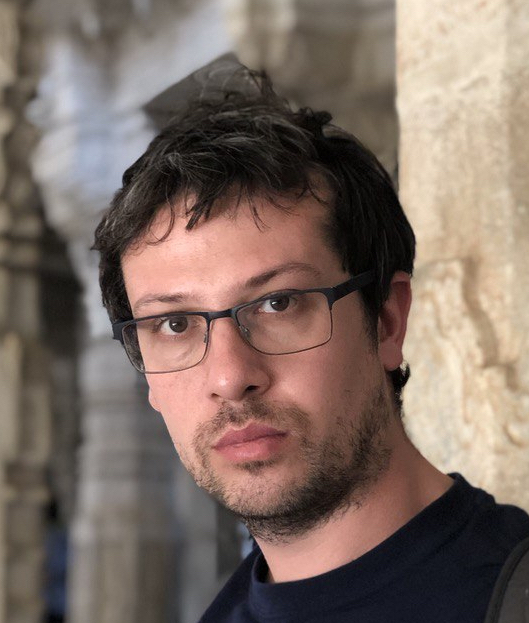}}]{Christian Pilato} 
is a Tenure-Track Assistant Professor at Politecnico di Milano. He was a Post-doc Research Scientist at Columbia University (2013-2016) and Università della Svizzera italiana (2016-2018). He was also a Visiting Researcher at New York University, TU Delft, and Chalmers University of Technology. He has a Ph.D. in Information Technology from Politecnico di Milano (2011). His research interests include high-level synthesis, reconfigurable systems and system-on-chip architectures, with emphasis on memory and security aspects. He served as program chair of EUC 2014 and is currently serving in the program committees of many conferences on EDA, CAD, embedded systems, and reconfigurable architectures (DAC, ICCAD, DATE, CASES, FPL, ICCD, etc.) He is a Senior Member of IEEE and ACM, and a Member of HiPEAC.
\end{IEEEbiography}

\begin{IEEEbiography}[{\includegraphics[width=1in,height=1.25in,clip,keepaspectratio]{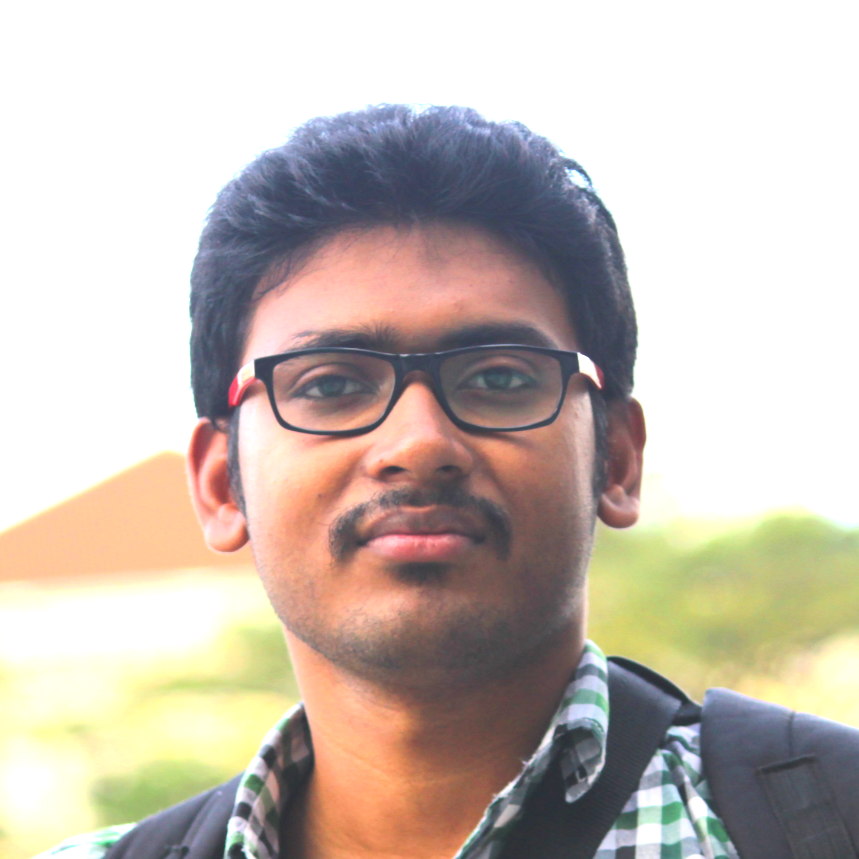}}]{Animesh Basak Chowdhury} 
received his MS in Computer Science from Indian Statistical Institute in 2016. Currently, he is a doctoral candidate at the NYU Centre for Cybersecurity. His research interests include Secure Electronics Design Automation (EDA), machine learning and SoC security. Prior to joining the Ph.D. program, he spent three years as a researcher at Tata Research Development and Design Centre (TRDDC), India, where he was primarily working in the area of formal verification and security testing. He has won several awards and recognition in International Software Verification and Testing Competitions (SV-COMP, TEST COMP, and RERS-Challenge).

\end{IEEEbiography}

\begin{IEEEbiography}
[{\includegraphics[width=1in,height=1.25in,clip,keepaspectratio]{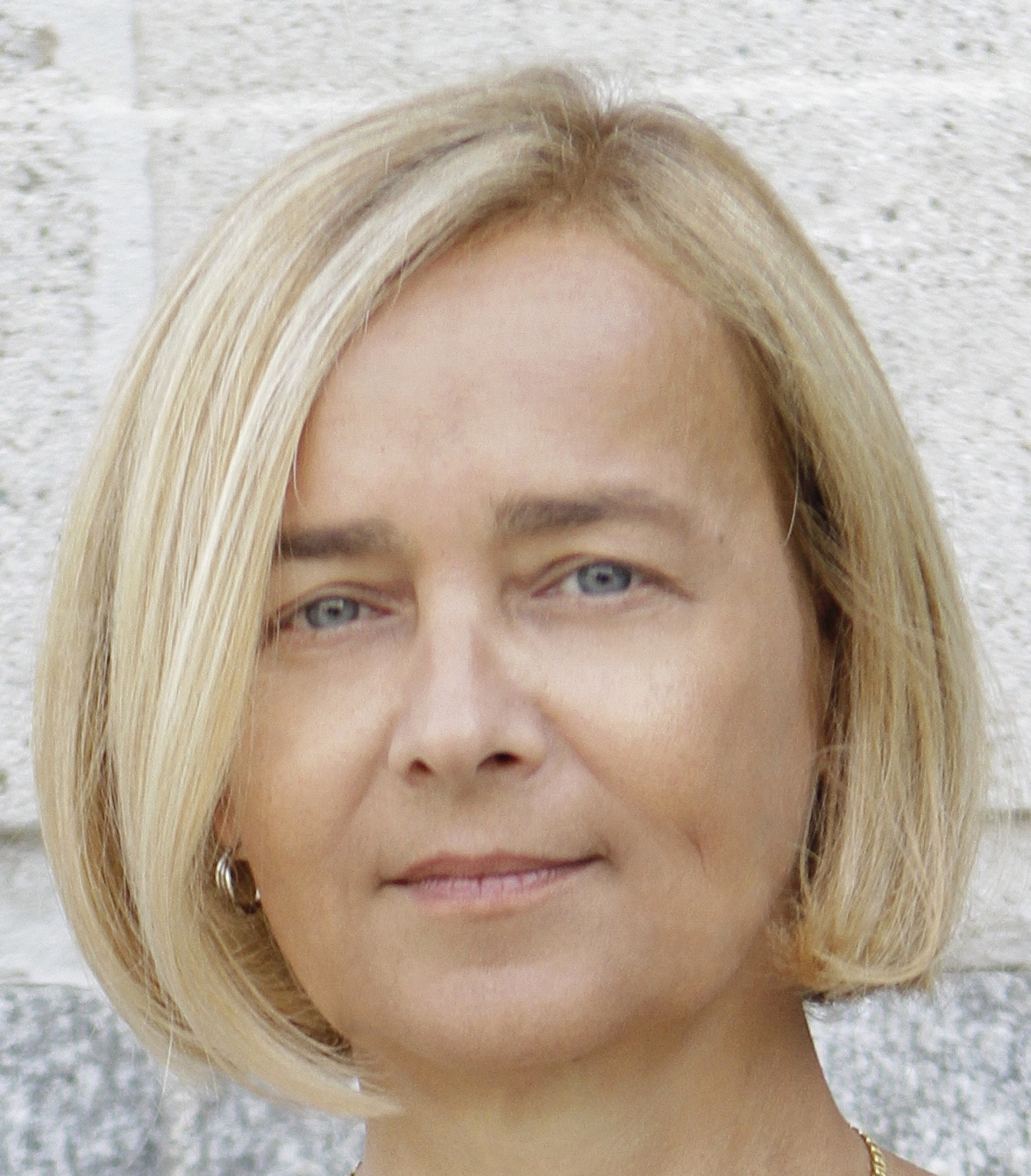}}]{Donatella Sciuto} 
received the  Laurea (Ms) in Electronic Engineering from Politecnico di Milano and the PhD in Electrical and Computer Engineering from the University of Colorado, Boulder, and the MBA from Bocconi University. 
She is currently the Executive Vice Rector of the Politecnico di Milano and Full Professor in Computer Science and Engineering. 
Her main research interests cover the methodologies for the design of embedded systems and multicore systems considering performance, power and security metrics. More recently she has been involved in managing and developing research projects in the area of smart cities and in the application of new ICT technologies to different application fields.  
She has published over 300 scientific papers.
She is a Fellow of IEEE for her contributions in embedded system design.  
She has served as Vice-President of Finance and then President of the IEEE Council of Electronic Design Automation from 2009 to 2013 and she serves in different capacities in IEEE Awards Committees, in scientific boards of IEEE journals and conferences. 
\end{IEEEbiography}

\begin{IEEEbiography}[{\includegraphics[width=1in,height=1.25in,clip,keepaspectratio]{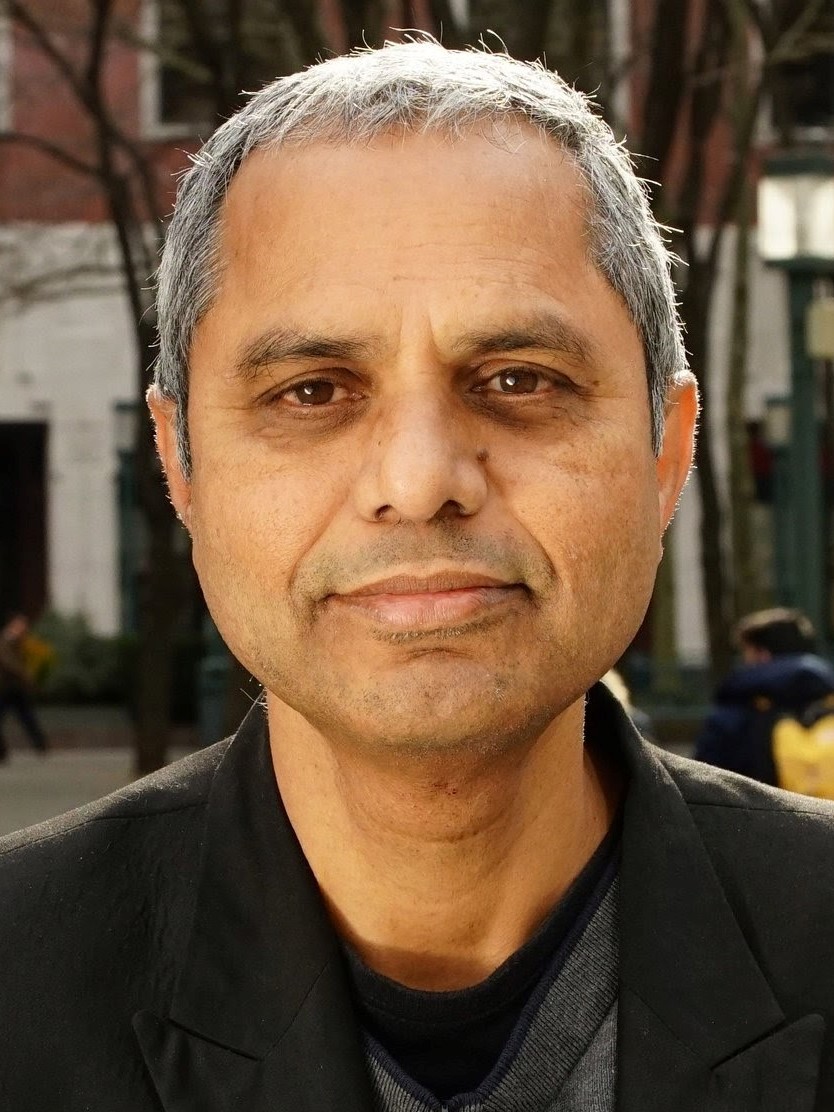}}]{Ramesh Karri} 
  is a Professor of ECE at New York University. He co-directs the NYU Center 
for Cyber Security (http://cyber.nyu.edu). He founded the Embedded Systems Challenge (https://csaw.engineering.nyu.edu/esc), the annual 
red team blue team event. He co-founded Trust-Hub (http://trust-hub.org). 
Ramesh Karri has a Ph.D. in Computer Science and 
Engineering, from the UC San Diego and a B.E in ECE from Andhra University. His 
research and education activities in hardware cybersecurity include trustworthy 
ICs; processors and cyber-physical systems; security-aware computer-aided 
design, test, verification, validation, and reliability; nano meets security; 
hardware security competitions, benchmarks, and metrics; biochip security; 
additive manufacturing security. He  published over 250 articles in leading 
journals and conference proceedings. Karri's work on hardware cybersecurity 
received best paper nominations (ICCD 2015 and DFTS 2015) and awards (ACM TODAES 
2018, ITC 2014, CCS 2013, DFTS 2013 and VLSI Design 2012). He received the 
Humboldt Fellowship and the NSF CAREER Award. He is the editor-in-chief of ACM JETC and serve(d)s  
on the editorial boards of IEEE and ACM Transactions (TIFS, TCAD, 
TODAES, ESL, D\&T, JETC). He was an IEEE Computer Society Distinguished 
Visitor (2013-2015). He served on the Executive Committee of the IEEE/ACM 
DAC leading the Security\@DAC initiative (2014-2017). 
He served as program/general chair of conferences and serves on program committees. 
He is a Fellow of the IEEE for leadership and contributions to Trustworthy Hardware.
\end{IEEEbiography}

\begin{IEEEbiography}
[{\includegraphics[width=1in,height=1.25in,clip,keepaspectratio]{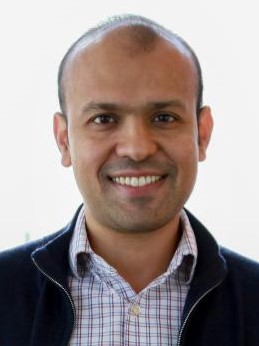}}]{Siddharth Garg} 
  received his Ph.D. degree in Electrical and Computer Engineering from 
Carnegie Mellon University in 2009, and a B.Tech. degree in Electrical 
Engineering from the Indian Institute of Technology Madras. He joined NYU in 
Fall 2014 as an Assistant Professor, and prior to that, was an Assistant 
Professor at the University of Waterloo from 2010-2014. His general research 
interests are in computer engineering, and more particularly in secure, reliable 
and energy-efficient computing. In 2016, Siddharth was listed in Popular 
Science Magazine's annual list of "Brilliant 10" researchers. Siddharth has 
received the NSF CAREER Award (2015), and paper awards at the IEEE Symposium on 
Security and Privacy (S\&P) 2016, USENIX Security Symposium 2013, at the 
Semiconductor Research Consortium TECHCON in 2010, and the International 
Symposium on Quality in Electronic Design (ISQED) in 2009. Siddharth also 
received the Angel G. Jordan Award from ECE department of Carnegie Mellon 
University for outstanding thesis contributions and service to the community. He 
serves on the technical program committee of several top conferences in the area 
of computer engineering and computer hardware, and has served as a reviewer for 
several IEEE and ACM journals.
\end{IEEEbiography}

\end{document}